\documentclass[letter]{IEEEtran}
\usepackage{amsmath,amsfonts,amssymb,amsthm,epsfig,epstopdf,url,array}
\usepackage{diagbox}
\usepackage{graphicx}
\usepackage{color}
\usepackage{textcomp,cite}
\usepackage[misc]{ifsym}
\usepackage{multirow}
\usepackage{caption,subcaption}
\usepackage{amsmath}

\newtheorem{theorem}{Theorem}

\newtheorem{lemma}{Lemma}

\newtheorem{remark}{Remark}

\begin{document}
\title{Performance Analysis and Optimal Design of HARQ-IR-Aided Terahertz Communications}
\author{Ziyang~Song,
        Zheng~Shi,
        Jiaji Su,
        Qingping Dou,
        Guanghua Yang,
        Haichuan~Ding,
        and Shaodan Ma
\thanks{\emph{Corresponding Author: Zheng Shi.}}
\thanks{Ziyang~Song, Zheng Shi, Jiaji Su, Qingping Dou, and Guanghua Yang are with the School of Intelligent Systems Science and Engineering, Jinan University, Zhuhai 519070, China (e-mails: 1670300710szy@stu2020.jnu.edu.cn; zhengshi@jnu.edu.cn; sjj123@stu2020.jnu.edu.cn; tdouqingping@jnu.edu.cn; ghyang@jnu.edu.cn).}
\thanks{Haichuan Ding is with the School of Cyberspace Science and Technology, Beijing Institute of Technology, Beijing 100081, China (email: hcding@bit.edu.cn).}
\thanks{Shaodan Ma is the State Key Laboratory of Internet of Things for Smart City, University of Macau, Macau, China (e-mail: shaodanma@um.edu.mo).}
}
\maketitle
\begin{abstract}
Terahertz (THz) communications are envisioned to be a promising technology for 6G thanks to its broad bandwidth. However, the large path loss, antenna misalignment, and atmospheric influence of THz communications severely deteriorate its reliability. To address this, hybrid automatic repeat request (HARQ) is recognized as an effective technique to ensure reliable THz communications. This paper delves into the performance analysis of HARQ with incremental redundancy (HARQ-IR)-aided THz communications in the presence/absence of blockage. More specifically, the analytical expression of the outage probability of HARQ-IR-aided THz communications is derived, with which the asymptotic outage analysis is enabled to gain meaningful insights, including diversity order, power allocation gain, modulation and coding gain, etc. Then the long term average throughput (LTAT) is expressed in terms of the outage probability based on renewal theory. Moreover, to combat the blockage effects, a multi-hop HARQ-IR-aided THz communication scheme is proposed and its performance is examined. To demonstrate the superiority of the proposed scheme, the other two HARQ-aided schemes, i.e., Type-I HARQ and HARQ with chase combining (HARQ-CC), are used for benchmarking in the simulations. In addition, a deep neural network (DNN) based outage evaluation framework with low computational complexity is devised to reap the benefits of using both asymptotic and simulation results in low and high outage regimes, respectively. This novel outage evaluation framework is finally employed for the optimal rate selection, which outperforms the asymptotic based optimization.

\end{abstract}
\IEEEpeerreviewmaketitle
\section{Introduction}\label{sec:int}

\IEEEPARstart
To meet the unprecedented increase of wireless devices and data traffic, there is an urgent need to boost the transmission rate. To accommodate the ever-growing demand, terahertz (THz) communications have emerged as a potential technique 
to support the peak data rate of more than 1~Tbps \cite{tataria20216g,9662195,9558848}. 
Different from the conventional low-frequency channel propagation characteristics, THz communications are susceptible to atmospheric turbulence, pointing errors, rain attenuation, etc. \cite{9787400,Kokkoniemi2015,Kokkoniemi2018,9530530,Priebe2012,9931325,Yacoub2007, Ekti2017,6998944}. As uncovered by \cite{9368251,9541155}
, these factors are harmful to the reception reliability of THz communications. To remedy this issue, it is important to come up with some methods for the reliability assurance of THz communications.

Recently, some efforts have been devoted to realize the reliable THz communications. More specifically, the reconfigurable intelligent surface (RIS) was proposed in \cite{9530530} to assist THz communications, where the RIS steers the reflected electromagnetic wave to the desired direction by adjusting the phase of the incident signal. In \cite{boulogeorgos2021coverage}, the connectivity analysis of RIS-assisted THz systems was carried out, and the results revealed that a minimum transmission power is required to guarantee 100\% coverage probability. The end-to-end SNR performance of THz communications over cascaded RISs was analyzed in \cite{9635819}. A stochastic approach was developed in \cite{9805798} to corroborate that the $\alpha$-$\mu$ distribution fits to the channel characterization of RIS-Aided THz communications. This then enabled the outage analysis, with which RIS-aided scheme was proved to be much more reliable than direct THz communications especially for a massive number of reflectors at the RIS. 
In \cite{9492775}, the relaying scheme with decode-and-forward (DF) mode was adopted to expand the coverage of THz communications in difficult-to-access terrain. In addition, as an important reliable transmission technique, hybrid automatic repeat request (HARQ) has been widely used in various communication systems. Thereupon, Type-I HARQ was employed to provide reliable THz communications in \cite{lopacinski2018implementation}, wherein its performance of bit error rate (BER) was studied through system-level simulations. To quantify the benefit of using HARQ, the outage performance of Type-I HARQ-aided and HARQ with chase combining (HARQ-CC)-aided THz communications was analyzed in \cite{song2021outage} from the information-theoretical perspective. 
Moreover, since HARQ with incremental redundancy (HARQ-IR) is superior to Type-I HARQ and HARQ-CC schemes due to its high coding complexity, it is necessary to scrutinize the performance of HARQ-IR-aided THz communications so as to substantially exploit the merits of HARQ. Thus, it motivates us to extend the analytical results in \cite{song2021outage} to the scenario of HARQ-IR-aided THz communications.

Nevertheless, unlike low frequency band, THz links are susceptible to be blocked by surrounding obstacles, thus yielding frequent communication interruptions. Clearly, the unfavorable impact of the blockage on THz communications cannot be disregarded. In \cite{9344846}, to capture the obstruction of walls and human bodies, etc., the coverage probability was studied by using the stochastic geometry. The analytical results revealed the detrimental impact of the node density, transmission distance, and atmosphere upon the coverage probability. Recently, some prior works have tried to solve this dilemma. In \cite{9127897}, massive multiple-input multiple-output (MIMO) aided THz communications were proposed for the blockage mitigation and a network sum-rate based optimization scheme was devised by considering only statistical channel state information (CSI). Moreover, in order to substantially overcome blockage impairment, one basic idea is altering the communication links so as to circumvent harsh propagation environment, which can be realized by combining relaying, reconfigurable intelligent surfaces (RIS), unmanned aerial vehicles (UAVs), etc. \cite{9083750,9743437,9672716,9370130}. More specifically, in \cite{9083750}, several relaying schemes were developed to address the blockage issue in THz communications, in which the performance gain from using relaying was rigorously evaluated. 
In \cite{9743437}, RIS was employed to aid THz communications and it is uncovered that the THz communications performance can be significantly enhanced by increasing the number of RIS's elements. Besides, a block coordinate searching algorithm was developed to jointly optimize the coordinates and phase shifts of RIS for THz communications in \cite{9672716}. Some other latest works \cite{9410457,9576697} also verified that integrating RIS into THz systems not only avoids the effect of blockage but also expands the coverage area. 
Moreover, in \cite{9370130}, the UAV was deployed to assist THz communications and the total transmission delay between the UAV and the ground user was minimized through optimizing the location, bandwidth and power, which can be solved with alternating algorithm. 
Unfortunately, most of prior optimization methodologies for designing THz communication systems have a prohibitively high computational complexity because of the complex fading characteristics of THz communications.

Motivated by the above issues, this paper thoroughly investigates the performance of HARQ-IR-aided THz communications in the presence/absence of blockage, with which the optimal system design is enabled. 
The major contributions of this paper are summarized as follows.
\begin{enumerate}
    \item  By considering $\alpha$-$\mu$ fading channels and pointing errors, the analytical outage and throughput expressions of HARQ-IR-aided THz communications are derived. The method of Abate and Whitt is then used to offer a highly accurate approximation for the outage expression. With the analytical results, the asymptotic outage analysis in the high SNR regime is performed to extract meaningful insights, include diversity order, power allocation gain, modulation and coding gain, etc.
  Additionally, it is worth noting that the outage analysis in this paper is more challenging than \cite{song2021outage} because of its complicated form of the accumulated mutual information by comparing to Type-I HARQ and HARQ-CC-aided schemes.
    \item  A multi-hop THz relaying scheme is proposed to mitigate the negative impact of blockage. To examine the impact of the blockage, both the outage probability and the long term average throughput (LTAT) of single-hop/multi-hop HARQ-IR-aided THz communications are obtained. 
    \item  A novel deep neural network (DNN) based outage evaluation framework with low computation complexity is devised to reap the benefits of using both asymptotic and simulation results in low and high outage regimes, respectively. This outage evaluation framework is then employed for the optimal rate selection. The numerical results indicate that the proposed framework surpasses the asymptotic based optimization in terms of the LTAT.
\end{enumerate}

In summary, the major novelties of this paper consist of providing an asymptotically analytical approach to capture the physical insights of HARQ-IR-aided THz communications, and developing a new DNN based outage evaluation framework to realize efficient THz system configurations.

The rest of this paper is structured as follows. Section \ref{sec:sys_mod} delineates the system model of HARQ-IR-aided THz communications. The outage probability and throughput of HARQ-IR-aided THz communications without blockage are analyzed in Section \ref{sec:opa}. Section \ref{sec:Rm} further extends the analytical results to HARQ-IR-aided THz communications with blockage, in which a multi-hop relaying scheme is proposed to circumvent obstacles. In Section \ref{sec:NR}, the simulation results validate our analysis. Besides, a deep neural network based outage evaluation framework is developed to combine the advantages of simulated and asymptotic results, with which the optimal rate selection is enabled. Finally, Section \ref{sec:con} concludes this paper.
\section{System Model}\label{sec:sys_mod}
\subsection{HARQ-IR-aided THz Communications}
As opposed to \cite{song2021outage} that two simple HARQ schemes (i.e., Type-I HARQ and HARQ-CC) are considered to assist THz communications, we employ the HARQ-IR scheme to further enhance the reliability. According to the HARQ-IR mechanism, each original information message contained $b$ information bits is first encoded to a long codeword, which is then chopped into $K$ sub-codewords with equal length $\ell$. Hence, a maximum of $K$ transmissions is allowable for delivering each message. The $K$ sub-codewords will be sent one by one till the successful decoding of the message. Moreover, at the receiver's side, all the incorrectly received sub-codewords are stored in buffer for the subsequent combination with the next received sub-codeword. Based on the success/failure of the reception of the message, a positive/negative acknowledgement (ACK) message will be fed back to the transmitter. At the transmitter, once either the ACK message is received or the maximum number of transmissions is reached, the transmission for the next message is initiated. 

By assuming blocking fading channels, the received signal ${{\bf y}_k}$ at the $k$-th HARQ round reads as
\begin{equation}\label{eqn:channel_model}
{{\bf y}_k} = \sqrt{P_k}{h_k}{\bf s}_k + {{\bf w}_k},\,  1 \le k \le K,
\end{equation}
where ${P_k}$, $h_k$, and ${\bf s}_k$ represent the transmit power, the THz channel coefficient, and the signal vector in the $k$-th HARQ round, respectively, ${{\bf w}_k}$ is the complex additive white Gaussian noise (AWGN) with variance $N_0$.

\subsection{$\alpha$-$\mu$ Fading Channels}
A variety of experimental measurements have verified that the $\alpha$-$\mu$ fading channel model fits well with the statistical characteristic of THz wave propagations. More specifically, the $\alpha$-$\mu$ fading channel model consists of three components, including the molecular absorption loss \cite{Kokkoniemi2015,Kokkoniemi2018}, pointing error \cite{9530530,Priebe2012,9931325}, and multipath fading \cite{Yacoub2007, Ekti2017}. In the meantime, by considering the versatility of $\alpha$-$\mu$ fading channels that encompass Rayleigh, Nakagami-m, Weibull, etc., as special cases, we employ $\alpha$-$\mu$ fading for the channel modeling of THz communications in this paper. 
To account for both large-scale and small-scale fadings, the channel coefficient $h_k$ is modeled as ${h_k} = {h_l}{h_{pf,k}}$ according to \cite{boulogeorgos2019analytical}, where $h_l$ stands for the deterministic THz path gain and ${h_{pf,k}}$ characterizes the combing effect of pointing errors and multipath fading. More precisely, ${h_l}$ is given by
\begin{equation}\label{eqn:path gain}
{h_l} = \frac{{c\sqrt {{G_{t}}{G_{r}}} }}{{4\pi {f}{d}}}\exp \left( { - \frac{1}{2}\kappa ({f},T,\psi ,p){d}} \right),
\end{equation}
where ${{G_{t}}}$ and ${{G_{r}}}$ denote the orientation dependent transmit and receive antenna gains, respectively \cite{balanis2011modern}, and higher antenna directivities are frequently needed in THz communications \cite{Priebe2012}; $c$, $f$, and $d$ denote the speed of light in free space, the carrier frequency, and the transmission distance, respectively; $\kappa ({f},T,\psi ,p)$ is the molecular absorption coefficient that is related to the temperature $T$, the relative humidity $\psi$ and the atmospheric pressure $p$. $\kappa ({f},T,\psi ,p)$ can be calculated by using \cite[eq. (4)]{boulogeorgos2019analytical}, which is omitted here to save space. For mathematical tractability, the maximum antenna gains $G_t$
and $G_r$ are assumed and the antenna misalignment is captured by the pointing error. The pointing error is possibly affected by various random factors, such as wind loads, thermal expansions, etc. \cite{farid2007}, which can be characterized by a compound distribution of Gaussian and Rayleigh \cite{farid2007,boulogeorgos2019analytical}. By combining the effects of pointing error and multipath fading, as derived in \cite{boulogeorgos2019analytical}, the probability density function (PDF) of ${h_{pf,k}}$ is given by
\begin{equation}\label{eqn:PDF}
{f_{\left| {{h_{pf,k}}} \right|}}(x) = \frac{{\phi {\mu ^{\frac{\phi }{\alpha }}}{x^{\phi  - 1}}}}{{S_0^\phi \hat h_f^\phi \Gamma (\mu )}}\Gamma \left( {\frac{{\alpha \mu  - \phi }}{\alpha },\frac{{\mu {x^\alpha }}}{{S_0^\alpha \hat h_f^\alpha }}} \right),
\end{equation}
where $ \Gamma \left( a,x \right)$ denotes the upper incomplete Gamma function, $\alpha$ and $\mu$ are fading parameters that provide flexibility to model various channel conditions, ${\hat h_f}$ denotes the $\alpha$-root mean value of the fading channel envelope. ${S_0} = {\left| {{\rm erf}(\zeta )} \right|^2}$ is the fraction of the collected power if both the transmit and receive antennas are fully aligned, where $\zeta  = \sqrt \pi  {r}/\left( {\sqrt 2 {w_{{d_0}}}} \right)$, $r$ and ${w_{{d_0}}}$ denote the radius of the receive antenna effective area and the transmission beam footprint radius at a reference distance ${d_0}$, respectively. $\phi  = w_e^2/4\sigma _s^2$ denotes the ratio of normalized beam-width to the jitter, where ${w_e}$ and ${\sigma _s}$ are the equivalent beam width radius and the doubled spatial jitter standard deviation, respectively, and $w_e^2 = \sqrt \pi{w_{{d_0}}^2  \rm erf\left( \zeta  \right)} /\left( {2\zeta \exp \left( { - {\zeta ^2}} \right)} \right)$.
\subsection{Blockage Modeling}
To capture the blockage effects of different obstacles, various blockage models have been proposed for THz communications \cite{9083750,9344846,9145081,8757158,9247469,9399124,9158415,9575003}. In general, the blockage model can be further classified into two categories from the spatial dimensional perspective, i.e., two-dimensional (2D) and three-dimensional (3D) blockage models. In particular, 2D blockage models were proposed in \cite{9083750,9344846} by overlooking the heights of the obstacles. Nevertheless, the 3D blockage modeling is mostly considered for THz communications. For example, a self-blocking model was constructed to account for the blockage from the user itself in \cite{9145081}. To characterize the human blocking effect for THz communications, an indoor human blockage model was proposed in \cite{8757158}. For the outdoor environment, the mobility of the human was taken into account for blockage modeling in \cite{9575003}. By considering the blockage from both the wall and the human, hybrid blockage models were built in \cite{9247469,9399124}. Moreover, the authors in \cite{9158415} proposed a UAV-oriented blockage model for drone communications.




%
%
%
%
%
%
%

As a representative common blockage model, the human blockage model developed in \cite{9247469} is employed in this paper, wherein the mobility of people follows the random directional model (RDM) and a human body can be modeled as a cylinder with the height ${h_b}$ and the radius ${r_b}$. The locations of people are modeled according to a homogeneous Poisson point process (HPPP) in the $x-y$ plane with intensity of ${\lambda _b}$. Besides, the heights of the BS's and the receiver's antennas are assumed to be ${h_a}$ and ${h_u}$, respectively.



For practicability, we assume that ${h_a} > {h_b} > {h_u}$, the non-blocking probability ${P_N}$ and the blocking probability ${P_B}$ can be calculated as \cite{9247469}
\begin{equation}\label{eqn:nbl}
{P_N} = 1 - {P_B}
= {e^{ - \beta d}},
\end{equation}
where $\beta  = 2{\lambda _b}{r_b}{({{h_b} - {h_u}})}/{({{h_a} - {h_u}})}$.

\section{HARQ-IR-Aided THz Without Blockage}\label{sec:opa}
In this section, a point-to-point HARQ-IR-aided THz communication system without blockage is considered. From the information-theoretical perspective \cite{caire2001throughput}, the equivalent accumulated mutual information per symbol of HARQ-IR-aided THz communications after $K$ HARQ rounds is given by
\begin{equation}\label{eqn:harq IR}
{I}_K = \sum\limits_{k = 1}^K {{{\log }_2}(1 + {\rho _k}{{\left| {{h_l}} \right|}^2}{{\left| {{h_{pf,k}}} \right|}^2})}.
\end{equation}
where ${\rho _k}$ refers to the transmit signal-to-noise ratio (SNR) in the $k$-th HARQ round, i.e. ${\rho _k} = {P_k}/{N_0}$. In addition, \eqref{eqn:harq IR} entirely differs from the formulations of the accumulated mutual information for the Type-I-aided and the HARQ-CC-aided schemes in \cite{song2021outage}, where the accumulated mutual information is expressed as a logarithm of the maximum or the summation of multiple random variables.

To evaluate the transmission reliability of HARQ-IR-aided THz communications, the outage probability is one of the paramount performance metrics \cite{Millimeter}. 
In particular, the outage probability is defined as the probability of the event that the accumulated mutual information is less than the initial transmission rate $R$ \cite{9234486}. Accordingly, the outage probability of HARQ-IR-aided scheme can be expressed as
\begin{equation}\label{eqn:pout}
p_{out,K}{\rm{ = Pr}}\left\{ {{I}_K < R} \right\},
\end{equation}
where $R = b/\ell$. In what follows, the outage analysis will be conducted and the asymptotic outage probability will be derived to explore useful insights. 
\subsection{Outage Analysis}
By substituting (\ref{eqn:harq IR}) into (\ref{eqn:pout}) along with the logarithmic property, \eqref{eqn:pout} can be derived as
\begin{align}\label{eqn:IR_1}
{p_{out,K}}&{\rm{ = Pr}}\left\{ {{{\log }_2}\left(\prod\limits_{k = 1}^K {(1 + {\rho _k}{{\left| {{h_l}} \right|}^2}{{\left| {{h_{pf,k}}} \right|}^2})} \right) < R} \right\}\notag\\
 &= \Pr \left\{ {\prod\limits_{k = 1}^K {\left(1 + {\rho _k}{{\left| {{h_l}} \right|}^2}{{\left| {{h_{pf,k}}} \right|}^2}\right)}  < {2^R}} \right\}.
\end{align}
Accordingly, the derivation of (\ref{eqn:IR_1}) boils down to determining the distribution of the product of multiple random variables. To this end, the Mellin transform can be utilized to derive (\ref{eqn:IR_1}) \cite{debnath2014integral}. Moreover, (\ref{eqn:IR_1}) can be expressed as
\begin{equation}\label{eqn:IR_2}
p_{out,K} = \frac{1}{{2\pi \rm i}}\int_{{{\rm{c}}_1} - \rm i\infty }^{{{\rm{c}}_1}{\rm{ + i}}\infty } {\frac{{{2^{ - Rt}}}}{{ - t}}\phi \left( {t + 1} \right)dt},
\end{equation}
where ${{{\rm{c}}_1}}  < 0$, $\rm i = \sqrt { - 1}$, $\phi \left( t \right)$ is the Mellin transform of the PDF of ${\prod\nolimits_{k = 1}^K {(1 + {\rho _k} {{\left| {{h_l}} \right|}^2}{{\left| {{h_{pf,k}}} \right|}^2})} }$, and $\phi \left( t \right)$ is obtained by
\begin{align}\label{eqn:IR_3}
\phi \left( t \right) =& \mathbb E\left\{ {{{\left( {\prod\limits_{k = 1}^K {(1 + {\rho _k} {{\left| {{h_l}} \right|}^2}{{\left| {{h_{pf,k}}} \right|}^2})} } \right)}^{t - 1}}} \right\} \notag\\
 =& \prod\limits_{k = 1}^K {\frac{{\phi {\mu ^{\frac{\phi }{\alpha }}}}}{{S_0^\phi \hat h_f^\phi \Gamma (\mu )}}\int_0^\infty  {{{(1 + {\rho _k} {{\left| {{h_l}} \right|}^2}x_k^2)}^{t - 1}}} } x_k^{\phi  - 1}\notag\\
 &\times \Gamma \left( {\frac{{\alpha \mu  - \phi }}{\alpha },\frac{{\mu x_k^\alpha }}{{S_0^\alpha \hat h_f^\alpha }}} \right)d{x_k},
\end{align}
where $\mathbb E\left\{ \cdot \right\}$ is the expectation operator. By plugging (\ref{eqn:IR_3}) into (\ref{eqn:IR_2}), it follows that
\begin{align}\label{eqn:IR_4}
&p_{out,K} = \frac{1}{{2\pi \rm i}}\int_{{{\rm{c}}_1} - \rm i\infty }^{{{\rm{c}}_1}{\rm{ + i}}\infty } {\frac{{{2^{ - Rt}}}}{{ - t}}\prod\limits_{k = 1}^K {\frac{{\phi {\mu ^{\frac{\phi }{\alpha }}}}}{{S_0^\phi \hat h_f^\phi \Gamma (\mu )}}}}\times\notag\\
& \int\nolimits_0^\infty  {{{(1 + {\rho _k} {{\left| {{h_l}} \right|}^2}x_k^2)}^t}}  x_k^{\phi  - 1}\Gamma \left( {\frac{{\alpha \mu  - \phi }}{\alpha },\frac{{\mu x_k^\alpha }}{{S_0^\alpha \hat h_f^\alpha }}} \right)d{x_k}dt.
\end{align}
By using the tool of Mellin transform and the representation of Fox's H-function, the outage probability of HARQ-IR-aided THz communications is derived in the following lemma.

\begin{lemma}\label{eqn:le_1}
$p_{out,K}$ can be expressed in the form of the inverse Laplace transform as (\ref{eqn:IR_9}), as shown at the top of the next page.
\begin{figure*}[!t]
\begin{equation}\label{eqn:IR_9}
{p_{out,K}} = {\left( {\frac{\phi }{{2\Gamma (\mu )}}} \right)^K}\frac{1}{{2\pi {\text{i}}}}\int_{ - {{\text{c}}_1} - {\text{i}}\infty }^{ - {{\text{c}}_1} + {\text{i}}\infty } {\underbrace {\frac{1}{t}\prod\limits_{k = 1}^K {\frac{{H_{3,2}^{1,3}\left[ {{{\left( {{\rho _k}{{\left| {{h_l}} \right|}^2}} \right)}^{\frac{1}{2}}}{{\left( {\frac{\mu }{{\hat h_f^\alpha S_{\text{0}}^\alpha }}} \right)}^{ - \frac{1}{\alpha }}}\left| {_{\left( {0,\frac{1}{2}} \right),\left( { - \phi ,1} \right)}^{\left( {1 - t,\frac{1}{2}} \right),\left( {1 - \mu ,\frac{1}{\alpha }} \right),\left( {1 + \phi ,1} \right)}} \right.} \right]}}{{\Gamma \left( t \right)}}} }_{F\left( t \right)}{e^{Rt\ln 2}}dt}
\end{equation}
\hrulefill
\end{figure*}
\end{lemma}
\begin{proof}
Please see Appendix \ref{sec:le1}.
\end{proof}

Needless to say, it is intractable to derive a closed-form expression for \eqref{eqn:IR_9} that involves a product of multiple Fox's H-function. Hence, by using the method of Abate and Whitt, (\ref{eqn:IR_9}) can be numerically computed with a high accuracy as 
 \begin{align}\label{eqn:lemma_1}
{p_{out,K}} \approx& {\left( {\frac{\phi }{{2\Gamma (\mu )}}} \right)^K}\frac{{{2^{ - M}}{e^{\frac{A}{2}}}}}{{R\ln 2}}\sum\limits_{m = 0}^M {\left( {\begin{array}{*{20}{c}}
  M \\
  m
\end{array}} \right)}\notag\\
&\times \left( \begin{gathered}
  \frac{1}{2}\Re \left\{ {F\left( {\frac{A}{{R\ln 4}}} \right)} \right\} \hfill \\
   + \sum\limits_{n = 1}^{Q + m} {{{\left( { - 1} \right)}^n}} \Re \left\{ {F\left( {\frac{{A + 2n\pi \text{i}}}{{R\ln 4}}} \right)} \right\} \hfill \\
\end{gathered}  \right),
\end{align}
where $M$ is the number of Euler summation terms and $Q$ is the truncation order. According to \cite{9502509}, there exists an approximate error between (\ref{eqn:IR_9}) and (\ref{eqn:lemma_1}), which is composed of the discretization error and the truncation error. Herein, the discretization error is bounded by $\left| \varepsilon  \right| \leqslant {e^{ - A}}/\left( {1 - {e^{ - A}}} \right)$. For example, to ensure a discretization error up to ${10^{ - 8}}$, $A$ is set to $A \approx 18$. Furthermore, the truncation error is tunable by choosing appropriate $M$ and $Q$. For instance, $M=11$ and $Q = 15$ are suggested in \cite{abate1995numerical}. By comparing to the outage evaluation through Monte Carlo simulations, \eqref{eqn:lemma_1} is more effective when calculating a low outage probability owing to its low complexity and high accuracy. Although \eqref{eqn:lemma_1} can provide a high calculation accuracy for the outage probability, it is almost impossible to extract useful insights from this complex expression. In what follows, we turn to the asymptotic outage analysis for meaningful insights.

\subsection{Asymptotic Outage Analysis}
The complex form of \eqref{eqn:IR_9} impedes the extraction of useful physical insights into the effect of fading parameters. This motivates us to resort to the asymptotic outage analysis in the high SNR regime. To ease the analysis, we assume that ${\rho _1}/q_1={\rho _2}/q_2= \cdots ={\rho _M}/q_M = \hat \rho$ in the sequel, that is, ${\rho _k} = {q_k}\hat \rho $ for $k\in [1,K]$, where $q_k$ denotes the power allocation factor for the $k$-th HARQ round. Accordingly, the asymptotic expression of (\ref{eqn:IR_9}) amounts to determining the asymptotic expression of the Fox's $\text H$-function in the high SNR regime, i.e., $\hat \rho  \to \infty $. By capitalizing on the complex analysis, the asymptotic expression of the outage probability of HARQ-IR-aided THz communications is given by the following theorem.
\begin{theorem}\label{eqn:le_2}
In the high SNR regime, the outage probability ${p_{out,K}}$ of HARQ-IR-aided THz communications is asymptotic to
\begin{align}\label{eqn:IR_1000}
{p_{out,K}} \simeq  {\cal A}  {{\cal L}\left( P \right)} \left( \mathcal C(R)\right)^{-\mathcal D} \triangleq {p_{out,K}^{asy}}
\end{align}
where ${\cal A}$ is the impact factor that combines the effects from the pointing error and fading channels together, and is expressed as
\begin{align}\label{eqn:A}
{\cal A} = 
\left\{ \begin{array}{l}
{\left( {\frac{{{\mu ^{\frac{\phi }{\alpha }}}\Gamma (\mu  - \frac{\phi }{\alpha })}}{{{{\left| {{h_l}} \right|}^\phi }\hat h_f^\phi S_{\rm{0}}^\phi \Gamma (\mu )N_0^{\frac{\phi }{2}}}}} \right)^K},\mu \alpha  - \phi  > 0\\
{\left( {\frac{{\phi {\mu ^{\mu  - 1}}}}{{\left( {\phi  - \mu \alpha } \right){{\left| {{h_l}} \right|}^{\mu \alpha }}\hat h_f^{\mu \alpha }S_{\rm{0}}^{\mu \alpha }\Gamma (\mu )N_0^{\frac{{\mu \alpha }}{2}}}}} \right)^K},\mu \alpha  - \phi  < 0
\end{array} \right.,
\end{align}
the term ${{\cal L}\left( \bf P \right)}$ refers to the impact factor from the power allocation, which is obtained as
\begin{align}\label{eqn:P}
{\cal L}\left( \bf P \right) = \left\{ \begin{array}{l}
{\left( {\prod\limits_{k = 1}^K {{P_k}} } \right)^{ - \frac{\phi }{2}}},\mu \alpha  - \phi  > 0\\
{\left( {\prod\limits_{k = 1}^K {{P_k}} } \right)^{ - \frac{{\mu \alpha }}{2}}},\mu \alpha  - \phi  < 0
\end{array} \right.,
\end{align}
${\bf P} =(P_1,\cdots,P_K)$, $\mathcal C(R) = {{{\cal G}_K}\left( {{2^R}} \right)}^{-1/\mathcal D} $ stands for the modulation and coding gain, and ${{\cal G}_K}{\left( {{2^R}} \right)}$ is expressed as \eqref{eqn:g_k}, as shown at the top of this page. 
\begin{figure*}[!t]
\begin{align}\label{eqn:g_k}
{{\cal G}_K}{\left( {{2^R}} \right)} =
\left\{ \begin{array}{l}
\left(\Gamma \left( {\frac{\phi }{2}+1} \right)\right)^K  G_{K,K}^{0,K}{\left( {{2^R}\left| {_{0,1,...,1}^{1 + \frac{\phi }{2},1 + \frac{\phi }{2},...,1 + \frac{\phi }{2}}} \right.} \right)},\quad \mu \alpha  - \phi  > 0\\
\left(\Gamma \left( {\frac{{\mu \alpha }}{2}+1} \right)\right)^K G_{K,K}^{0,K}{\left( {{2^R}\left| {_{0,1,...,1}^{1 + \frac{{\mu \alpha }}{2},1 + \frac{{\mu \alpha }}{2},...,1 + \frac{{\mu \alpha }}{2}}} \right.} \right)},\quad \mu \alpha  - \phi  < 0
\end{array} \right.,
\end{align}
\hrulefill
\end{figure*}
${{\cal D}}$ denotes the diversity order and is given by
\begin{align}\label{eqn:d_I}
{{\cal D}} = \frac{K\min\{\phi,\alpha \mu\}}{2}.
\end{align}
\end{theorem}
\begin{proof}
Please see Appendix \ref{sec:le2}.
\end{proof}

To explore more insights, the above relevant terms are carefully scrutinized one by one.

\subsubsection{Impact Factor of Pointing Error and Fading Channels ${\cal A}$}
The factor ${\cal A}$ comprises the impacts of the pointing error and the fading channels. Clearly from \eqref{eqn:A}, the value of the impact factor ${\cal A}$ depends on the sign of ${\mu \alpha  - \phi }$. More specifically, if ${\mu \alpha  - \phi } > 0$, the value of \eqref{eqn:A} is mainly affected by the pointing error. Otherwise, the fading channels dominate the value of the impact factor ${\cal A}$.



\subsubsection{Impact Factor of Power Allocation ${{\cal L}\left( \bf P \right)}$}
The impact factor ${{\cal L}\left( \bf P \right)}$ quantifies the effect of the power allocation scheme upon the outage performance. Similarly, it is found that the sign of ${\mu \alpha  - \phi }$ also determines the value of ${{\cal L}\left( \bf P \right)}$. Nonetheless, no matter whether ${\mu \alpha  - \phi }$ is less than or larger than zero, the asymptotic outage probability is always an inverse function of the product of the transmission powers in all HARQ rounds. This special feature may facilitate the optimal power allocation of HARQ-IR-aided THz communications with some off-the-shelf optimization tools, such as geometric programming.



\subsubsection{Modulation and Coding Gain ${{\cal C}}(R)$}\label{sec:cr}
The modulation and coding gain ${{\cal C}}(R)$ quantifies the amount of the SNR that can be reduced by the modulation and coding scheme (MCS) to achieve a certain outage target. Clearly from \eqref{eqn:g_k}, the explicit expression of ${{\cal C}}(R)$ depends on the sign of ${\mu \alpha  - \phi }$. Different from \cite{7959548} that the fading channels dominate the modulation and coding gain ${{\cal C}}(R)$, both the fading channels and the pointing error affect ${{\cal C}}(R)$ for HARQ-IR aided THz communications. Although ${{\cal C}}(R) = {{{\cal G}_K}\left( {{2^R}} \right)}^{-1/\mathcal D} $ is a complex function of the transmission rate $R$, it is proved in \cite{7959548} that ${{{\cal G}_K}\left( {{2^R}} \right)}$ is an increasing and convex function of $R$.

\subsubsection{Diversity Order $\mathcal D$}
The diversity order is a fundamental asymptotic reliability metric to characterize the degree of freedom for a communication system. More precisely, the diversity order $\mathcal D$ is defined as the slope of the outage probability against the transmit SNR on a log-log scale as \cite{shi2019achievable}
\begin{equation}\label{eqn:div}
{{\cal D}} =  - \mathop {\lim }\limits_{{\hat \rho}  \to \infty } \frac{{\log  {{p_{out,K}}}}}{{\log {\hat \rho }  }}.
\end{equation}

Clearly, the diversity order is linearly proportional to the maximum number of transmissions, i.e, $K$. This means that the full time diversity is achievable by the HARQ-IR-aided scheme. Moreover, by comparing \eqref{eqn:d_I} to the results in \cite{song2021outage}, the same diversity order can be attained as the Type-I HARQ and the HARQ-CC-aided schemes. However, it is worth highlighting that the HARQ-IR-aided scheme performs the best among them owing to its highest coding and modulation gain ${{\cal C}}(R)$, which will be verified in the simulation section. Furthermore, the diversity order is also constrained by the severity of both the pointing errors and $\alpha$-$\mu$ fading. The asymptotic performance metrics of three types of HARQ-aided THz communication schemes are tabulated in Table \ref{tab:sum}, as shown at the top of the next page, wherein the asymptotic results of Type-I HARQ and HARQ-CC-aided THz communications can be found in \cite{song2021outage}.

\begin{figure*}[!t]
\centering
\captionsetup{type=table} 
\caption{The asymptotic performance metrics of HARQ-aided THz communication schemes.}
\label{tab:sum}%
\begin{tabular}{|c||c|c|c|c|}
 \hline
 \diagbox[width=5.5cm]{HARQ-Aided Schemes}{Metrics}& $\mathcal A$ & $\mathcal L({\bf P})$ & $\mathcal C(R)$  & $\mathcal D$ \\
 \hline
 \hline
 Type-I HARQ  &  \multirow{3}{*}{\eqref{eqn:A}}  & \multirow{3}{*}{\eqref{eqn:P}} & $\frac{1}{2^R-1}$ & \multirow{3}{*}{\eqref{eqn:d_I}}\\ \cline{1-1} \cline{4-4}

 HARQ-CC  &    &  & ${\left( {\Gamma \left( {\frac{\mathcal D}{K} + 1} \right)} \right)^{ - {K}/{\mathcal D}}}{\left( {\Gamma \left( {\mathcal D + 1} \right)} \right)^{{1}/{\mathcal D}}}\frac{1}{{{2^R} - 1}}$ & \\ \cline{1-1} \cline{4-4}

 HARQ-IR  &    &  & ${{{\cal G}_K}\left( {{2^R}} \right)}^{-1/\mathcal D}$ & \\ \cline{1-1} \cline{4-4}
 \hline
\end{tabular}
\end{figure*}


\subsection{Long Term Average Throughput}
In this paper, the long term average throughput (LTAT) is used to evaluate the throughput performance of the HARQ-IR-aided THz communication systems. More specifically, LTAT is defined as the ratio of the total number of successfully received bits to the consumed time, that is, \cite{6879267}
\begin{equation}\label{eqn:t_def}
 \bar {\mathcal T} = \mathop {\lim }\limits_{t \to \infty } \frac{{b_t^{suc}}}{t},
\end{equation}
where $b_t^{suc}$ denotes the total number of successfully received information bits till time $t$. By invoking the renewal reward theory, the LTAT of HARQ-aided schemes can be obtained as
\begin{align}\label{eqn:LTAT_10}
{\bar {\cal T}} = \frac{{R\left( {1 - {p_{out,K}}} \right)}}{{{\bar {\cal N}}}},
\end{align}
where ${\bar {\cal N}}$ is the average number of transmissions and it is given by
\begin{equation}\label{eqn:avg_num}
{\bar {\cal N}} = 1 + \sum\nolimits_{k = 1}^{K - 1} {p_{out,k}}.
\end{equation}

\section{HARQ-IR-Aided THz With Blockage}\label{sec:Rm}
Due to the ultra-high frequency of THz communications, the THz communication link is easily blocked by surrounding obstacles. In this section, the blockage effect is considered into the performance analysis of HARQ-IR-aided THz communications. Moreover, the relaying scheme is adopted to assist the HARQ-IR-aided THz communications to avoid the communication interruption caused by the blockage effect.

\subsection{Outage Analysis}
In what follows, we first consider a single-hop HARQ-IR-aided THz communication system. Then multiple relays are deployed to assist HARQ-IR-aided THz communications.
\subsubsection{Single-Hop}
In the single-hop HARQ-IR-aided THz communication system, the outage event occurs in the following two cases. For the first one, the accumulated mutual information is below the preset transmission rate if there is no blocking. For the second one, the THz communication is interrupted by the blockage. 
Therefore, by using the law of total probability, the outage probability of single-hop HARQ-IR-aided THz communication system after $K$ HARQ rounds is given by
\begin{equation}\label{eqn:nbloutK}
p_{out,K}^{\left( 1 \right)} = {P_B} + {P_N}{p_{out,K}},
\end{equation}
where ${P_B}$ and ${P_N}$ denote the blocking and non-blocking probability as given by \eqref{eqn:nbl}. 

\subsubsection{Multi-Hop}
\begin{figure}
        \centering
        \includegraphics[width=3in,trim=1 1 1 1,clip]{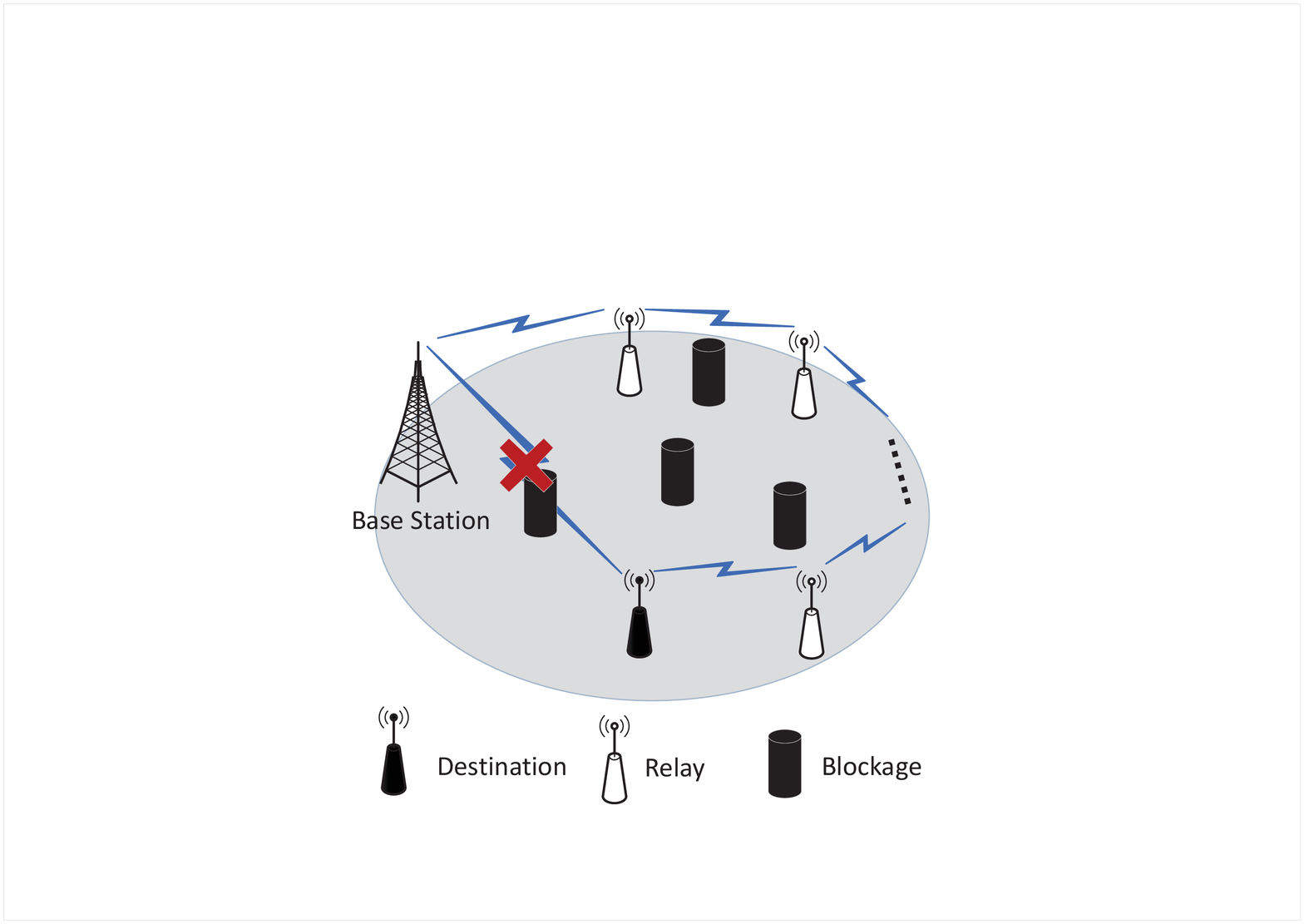}
        \caption{A multi-hop THz communication networks with blockage}\label{fig:R10}
\end{figure}
In this subsection, a multi-hop relaying scheme is proposed to overcome harsh propagation environments of THz communications, as shown in Fig. \ref{fig:R10}. 
In Fig. \ref{fig:R10}, we assume that the total number of hops is $L$ and the distance of the $l$-th hop is defined as ${d_l}$, where $1 \leqslant l \leqslant L$. According to \eqref{eqn:nbl}, the blocking effect in multi-hop THz communications can be mitigated because of $d_l < d$, where $d$ is the distance between the source and the destination. 
Moreover, the maximum number of transmissions in the multi-hop HARQ-IR-aided THz system is also assumed to be $K$.
The outage event of multi-hop relaying scheme takes place if any hops fail to decode the message within a total number of $K$ transmissions. By considering the complementary event of the outage, the outage probability of the multi-hop HARQ-IR-aided THz communications can be obtained as
\begin{align}\label{eqn:MH_1}
p_{out,K}^{\left( L \right)} = 1 - \sum\limits_{k = L}^K {p_{suc,k}^{\left( L \right)}},
\end{align}
where $p_{suc,k}^{\left( L \right)}$ denotes the probability of the event that the receiver successfully received the message at the $k$-th transmission. It is noteworthy that the subscript $k$ should start from $L$. This is due to the fact that each hop needs at least one transmission for the successful delivery of the message. To simplify \eqref{eqn:MH_1}, we assume that the number of transmissions for the successful reception at the $k$-th hop is denoted by ${\kappa_l}$. 
Clearly, the sum of the numbers of transmissions at all hops should not exceed the maximum number of transmissions, i.e. $\sum\nolimits_{l = 1}^L {{\kappa_l} \le K} $, and ${\kappa_l} \ge 1$. Moreover, the successful decoding event of the multi-hop HARQ-IR-aided THz system occurs if and only if the decoding of each hop is successful. Accordingly, $p_{suc,k}^{\left( L \right)}$ can be derived as
\begin{equation}\label{eqn:psuc_def}
p_{suc,k}^{\left( L \right)} = \sum\limits_{\sum\nolimits_{l = 1}^L {{\kappa _l}} = k ,{\kappa _1}, \cdots ,{\kappa _L} \ge 1} {\prod\limits_{l = 1}^L {P_N^l{p^l_{suc,{\kappa _l}}}} },
\end{equation}
where ${{p^l_{suc,{\kappa _l}}}}$ represents the successful probability in the ${{\kappa _l}}$-th HARQ round at the $l$-th hop in the absence of blockage and ${P_N^l}$ denotes the non-blocking probability at the $l$-th hop. By substituting \eqref{eqn:psuc_def} into \eqref{eqn:MH_1}, we have
%
\begin{align}\label{eqn:MH_2}
p_{out,K}^{\left( L \right)} = 1 - \sum\limits_{L \le \sum\nolimits_{l = 1}^L {{\kappa _l}} \le K,{\kappa _1}, \cdots ,{\kappa _L} \ge 1} {\prod\limits_{l = 1}^L {P_N^l{p^l_{suc,{\kappa _l}}}} }.
\end{align}
Furthermore, with regard to the successful probability ${{p^l_{suc,{\kappa _l}}}}$, the successful event at the ${\kappa _l}$-th transmission implies that the $l$-th hop fails to decode the message prior to the ${\kappa _l}$-th transmission. Therefore, the successful probability ${p^l_{suc,{\kappa _l}}} $ can be expressed as ${p^l_{suc,{\kappa _l}}} = {p^l_{out,{\kappa _l} - 1}} - {p^l_{out,{\kappa _l}}}$, where ${p^l_{out,{\kappa _l}}}$ refers to the outage probability after ${\kappa _l}$ transmissions at the $l$-th hop in the absence of blockage, and we stipulate ${p^l_{out,0}}=1$. It is worthwhile to note that the outage probability ${p^l_{out,{\kappa _l}}}$ can be calculated by using the results in Section \ref{sec:opa}. With this equation, the outage probability of the multi-hop HARQ-IR-aided THz communications can be finally rewritten as
\begin{multline}\label{eqn:nbll}
p_{out,K}^{\left( L \right)} = 1 - \\
\sum\limits_{L \le \sum\nolimits_{l = 1}^L {{\kappa _l}} \le K ,{\kappa _1}, \cdots ,{\kappa _L} \ge 1} {\prod\limits_{l = 1}^L {P_N^l\left( {p_{out,{\kappa _l} - 1}^l - p_{out,{\kappa _l}}^l} \right)} }.
\end{multline}

On the basis of \eqref{eqn:nbll}, the asymptotic outage analysis of multi-hop HARQ-IR-aided THz communications can be conducted to gain useful insights, as summarized in the following remark.
\begin{remark}\label{eqn:mh_do}
In the high SNR regime, i.e., $\hat \rho \to \infty$, there exists a lower bound of the outage probability if $P_B>0$, i.e.,
\begin{align}\label{mh_do2}
p_{out,K}^{\left( L \right)} &\simeq 1 - {P_N}^L,\,K\ge L.
\end{align}
Whereas, if there is no blockage, i.e., $P_B=0$, the outage probability converges to zero as $\hat \rho \to \infty$, and the diversity order of multi-hop HARQ-IR-aided THz communications without blockage is given by
\begin{equation}\label{mh_do3}
{{\cal D}} = \frac{{\left( {K - L + 1} \right)\min \{ \phi ,\alpha \mu \} }}{2},\,K\ge L.
\end{equation}
Clearly, \eqref{mh_do3} includes the diversity order of the single-hop case in \eqref{eqn:d_I} as a special case. Moreover, the number of hops, i.e., $L$, decreases the diversity order, consequently degrades the outage performance. Full time diversity is unreachable if $L>1$. Hence, it does not mean that more hops will be in favor of the reliable transmissions.
\end{remark}
\begin{proof}
Please see Appendix \ref{sec:le3}.
\end{proof}


\subsection{LTAT}
In the following, the LTAT of HARQ-IR-aided THz systems with the blockage is investigated by considering two cases, i.e., single-hop and multi-hop.
\subsubsection{LTAT of Single-Hop}
Similarly to \eqref{eqn:LTAT_10}, the LTAT of the single-hop HARQ-IR-aided THz system with blockage is given by
\begin{align}\label{eqn:LTAT_1}
{{{\bar {\cal T}}}^{\left( 1 \right)}} = \frac{{R\left( {1 - p_{out,K}^{\left( 1 \right)}} \right)}}{{{{{\bar {\cal N}}}^{\left( 1 \right)}}}},
\end{align}
where ${{{\bar {\cal N}}}^{\left( 1 \right)}} = 1 + \sum\nolimits_{k = 1}^{K - 1} {p_{out,k}^{\left( 1 \right)}}  $ corresponds to the average number of transmissions.

\subsubsection{LTAT of Multi-Hop}
Similarly, the LTAT of the multi-hop HARQ-IR-aided THz system with the blockage is also expressed as
\begin{align}\label{eqn:LTAT_2}
{{{\bar {\cal T}}}^{\left( L \right)}} = \frac{{R\left( {1 - p_{out,K}^{\left( L \right)}} \right)}}{{{{{\bar {\cal N}}}^{\left( L \right)}}}},
\end{align}
where ${{{\bar {\cal N}}}^{\left( L \right)}}$ represents the average number of transmissions in the case of $L$ hops. By considering the distribution of the number of transmissions, the expectation of the number of transmissions can be obtained as
\begin{align}\label{eqn:LTAT_3}
{{{\bar {\cal N}}}^{\left( L \right)}} = Kp_{out,K - 1}^{\left( L \right)} + \sum\limits_{k = L}^{K - 1} {kp_{suc,k}^{\left( L \right)}},
\end{align}
where the first term indicates the maximum number $K$ of transmissions and the subscript $k$ starts from $L$ because at least $L$ transmissions are required for a successful delivery. Moreover, the successful probability can be represented by the difference of the probabilities of two successive outage events, i.e., $p_{suc,k}^{\left( L \right)} = {p_{out,k - 1}^{\left( L \right)} - p_{out,k}^{\left( L \right)}}$.
Accordingly, \eqref{eqn:LTAT_3} can be rewritten as
\begin{align}\label{eqn:LTAT_4}
{{{\bar {\cal N}}}^{\left( L \right)}} = Kp_{out,K - 1}^{\left( L \right)} + \sum\limits_{k = L}^{K - 1} {k\left( {p_{out,k - 1}^{\left( L\right)} - p_{out,k}^{\left( L \right)}} \right)}.
\end{align}
After some algebraic manipulations, the average number of transmissions in the case of multi-hop can be eventually obtained as
\begin{align}\label{eqn:LTAT_5}
{{{\bar {\cal N}}}^{\left( L \right)}} = L + \sum\limits_{k = L}^{K - 1} {p_{out,k}^{\left( L \right)}}.
\end{align}

Interestingly, a direct consequence of Remark \ref{eqn:mh_do} yields the asymptotic expression of the LTAT, as concluded in the following remark.
\begin{remark}
In the high SNR regime, i.e., $\hat \rho \to \infty$, there exists an upper bound of the LTAT if $P_B>0$, i.e.,
\begin{align}\label{mh_do10}
{{{\bar {\cal T}}}^{\left( L \right)}} &\simeq \frac{R}{{K\left( {{P_N}^{ - L} - 1} \right) + L}}.
\end{align}
Whereas, if there is no blockage, i.e., $P_B=0$, the LTAT converges to $R/L$ as $\hat \rho \to \infty$. Apparently, more hops degrade the asymptotic LTAT no matter $P_B>0$ or not.
\end{remark}

\section{Numerical Results}\label{sec:NR}
In this section, the analytical results are verified by conducting Monte Carlo simulations. For illustration, the system parameters are set as follows, Rayleigh fading channels are assumed such that $\alpha  = 2$ and $\mu  = 1$, and standard environmental conditions $\psi  = 50\% $, $T = 296{\rm{^\circ }}$~K, and $p = 101325$~Pa are considered. In addition, we assume that $f=275$~GHz, $K=3$, $R = 2$~bps/Hz, and ${\sigma _{\rm{s}}} = 1$ unless otherwise specified. Moreover, according to \cite{Guan2017}, both the transmit and the receive antenna gains are set to $55$~dBi, meanwhile the total transmission distance is assumed to be $20$~m. Besides, the number of Monte Carlo simulation runs is set up to ${10^{7}}$. By considering that the sign of $\alpha \mu  - \phi$ has different impacts on the outage performance, we consider the following two cases: 1) the equivalent beamwidth ${w_{{d_0}}}$ is set as $1$ for $\alpha \mu  - \phi  > 0$; 2) ${w_{{d_0}}}$ is set as $3$ for $\alpha \mu  - \phi  < 0$. By employing fixed power allocation scheme, the power allocation coefficients are set as $q_1=\cdots=q_K=1$. Furthermore, the curves of the simulated, the analytical, and the asymptotic results in the following figures are labeled as ``Sim.'', ``Ana.'', and ``Asy.'', respectively, where the analytical results are obtained by adopting the method of Abate and Whitt as given in \eqref{eqn:IR_101}. For comparison, the results for Type-I HARQ-aided and HARQ-CC-aided schemes in \cite{song2021outage} are incorporated for benchmarking purpose.

\subsection{Without Blockage}
Fig. \ref{fig:R1} depicts the outage probability ${p_{out,K}}$ versus the average transmit SNR $\hat \rho$ for different HARQ-aided THz communication schemes in the absence of blockage. It is clearly seen that the analytical and simulated results are in perfect agreement. 
One can also observe that the asymptotic curves tightly approach to the analytical curves with the increase of SNR no matter whether $\alpha \mu  - \phi  > 0$ or not, which validates the asymptotic analysis. Nonetheless, it is noteworthy that there exists a large gap between the analytical and asymptotic results for evaluating medium-to-high outage probabilities. Furthermore, it is found that the slopes of the outage curves of the three types of HARQ-aided schemes are identical. This justifies the validity of the diversity order analysis. In particular, the diversity order becomes independent of the fading parameters if ${\mu \alpha  - \phi < 0 }$. Otherwise, the diversity order of the system is determined by the pointing error. From Fig. \ref{fig:R1}, it is revealed that the HARQ-IR-aided scheme outperforms the other two HARQ-aided schemes, which corroborates the superiority of the HARQ-IR-aided scheme. This is because that the HARQ-IR aided scheme can achieve the highest modulation and coding gain by comparing to the other two HARQ-aided schemes, as shown in Table \ref{tab:sum}.


Fig. \ref{fig:R2} depicts the LTAT ${\bar {\cal T}}$ versus the average transmit SNR $\hat \rho$ for different HARQ-aided THz communication schemes in the absence of the blockage. Clearly, there is a perfect agreement between the simulation results and the analytical results no matter for $\alpha \mu  - \phi  > 0$ or not, which justifies the validity of our analysis. Not surprisingly, the LTAT can be improved by increasing the average SNR and is upper bounded by the transmission rate $R=2$~bps/Hz. Moreover, it can be seen that the HARQ-IR-aided scheme performs the best among the three HARQ-aided schemes especially for the low-to-medium SNR. This is due to the highest spectral efficiency of the HARQ-IR scheme, which attributes to its highest coding complexity.
\begin{figure}
        \centering
        \includegraphics[width=3.5in]{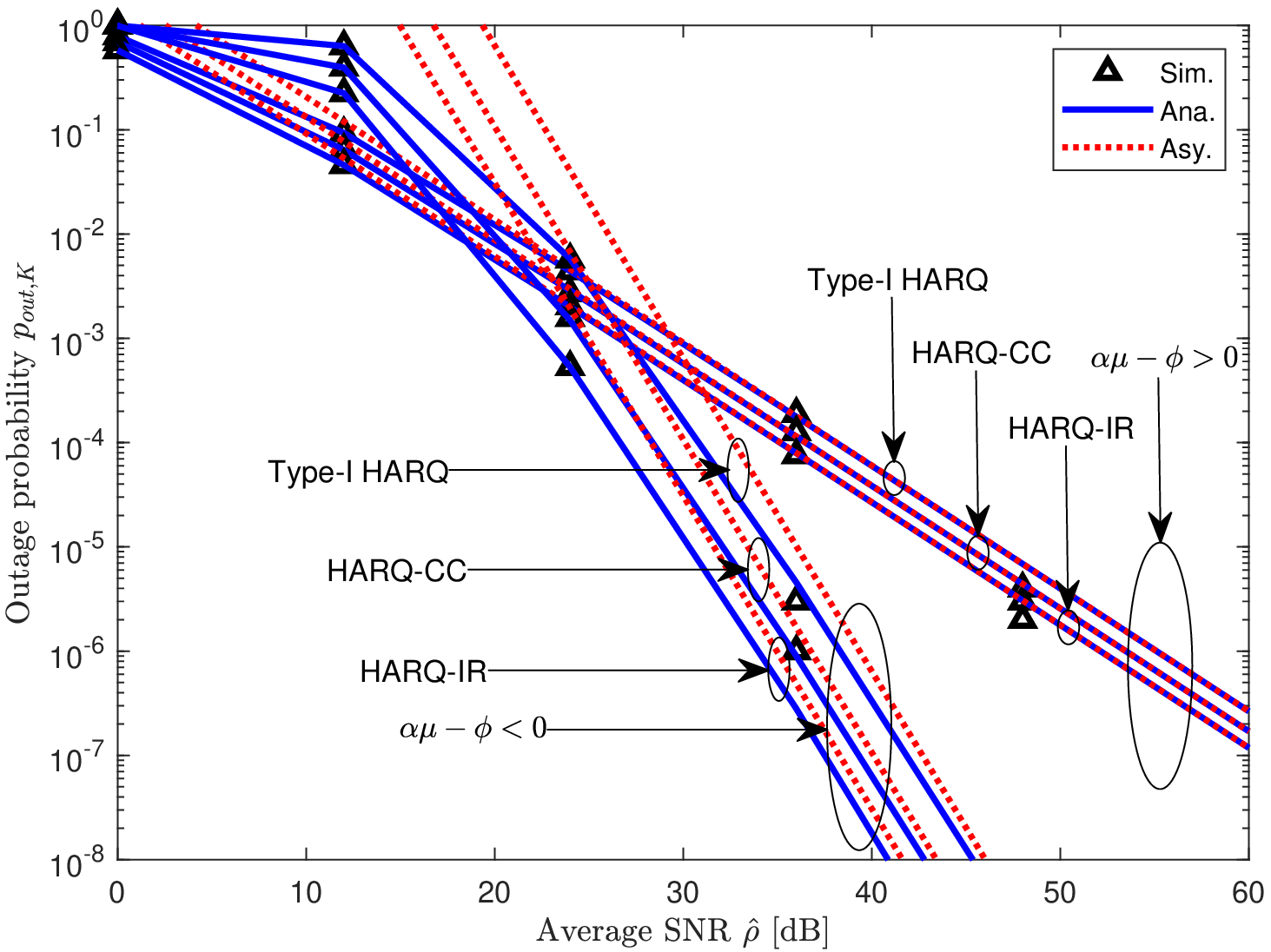}
        \caption{The outage probability $p_{out,K}$ versus the average SNR $\hat \rho $.}\label{fig:R1}
\end{figure}
\begin{figure}
        \centering
        \includegraphics[width=3.5in]{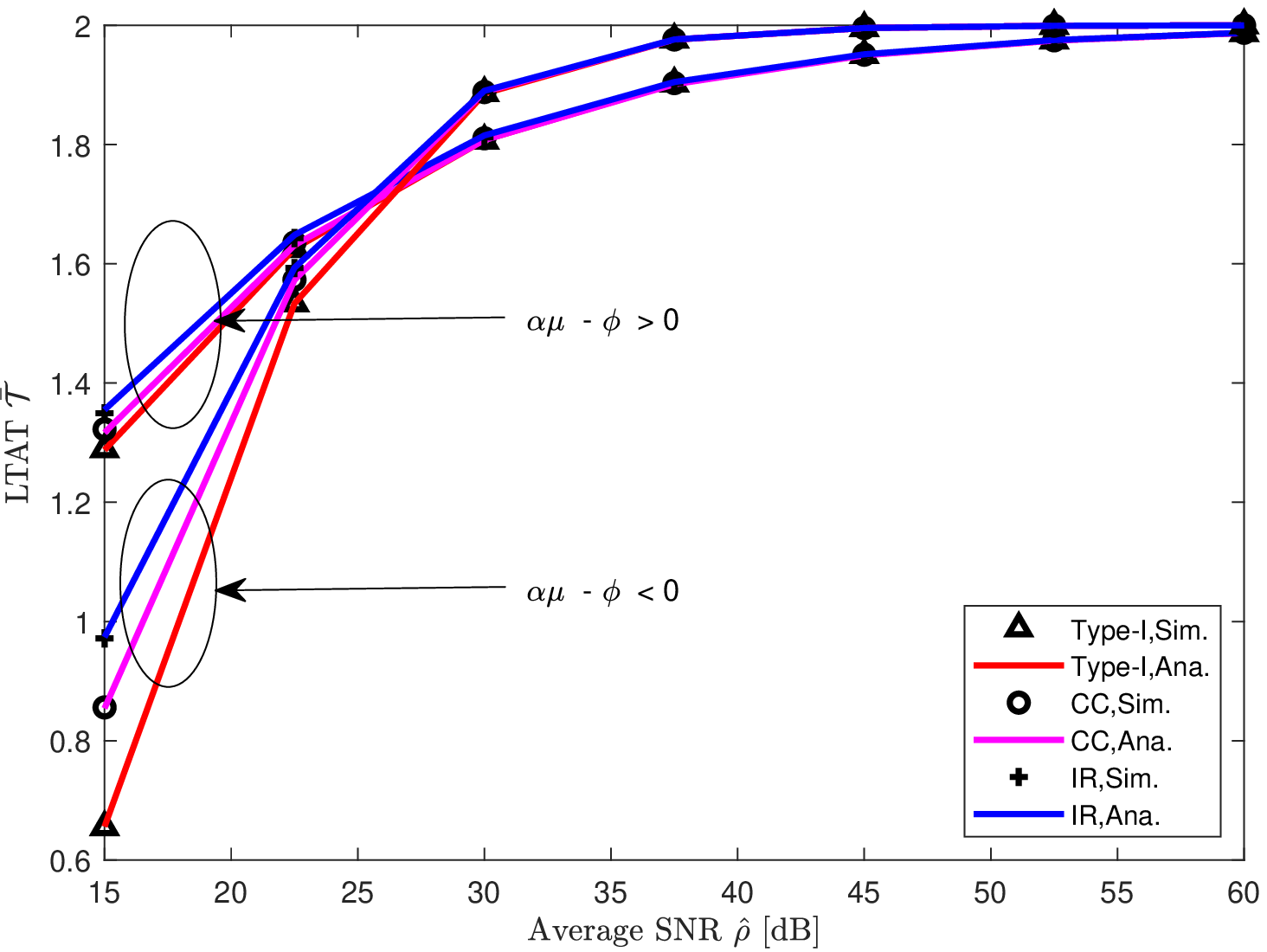}
        \caption{\textcolor[rgb]{0.00,0.00,0.00}{The LTAT ${\bar {\cal T}}$ versus the average SNR $\hat \rho $ }.}\label{fig:R2}
\end{figure}

Figs. \ref{fig:new1} and \ref{fig:R12d} are plotted for different carrier frequencies ($f=275,~350,~400$ GHz) and transmission distances ($d=10,~20, ~30$ m), respectively. It can be observed that the increase of the carrier frequency has a negative impact on the outage performance. This coincides with our intuition that a high carrier frequency usually yields a large path loss. Moreover, it is not beyond our expectation that the increase of the transmission distance deteriorates the outage performance. More importantly, it is shown that different carrier frequencies and transmission distances only affect the modulation and coding gain without altering the diversity order.
\begin{figure}
        \centering
        \includegraphics[width=3.5in]{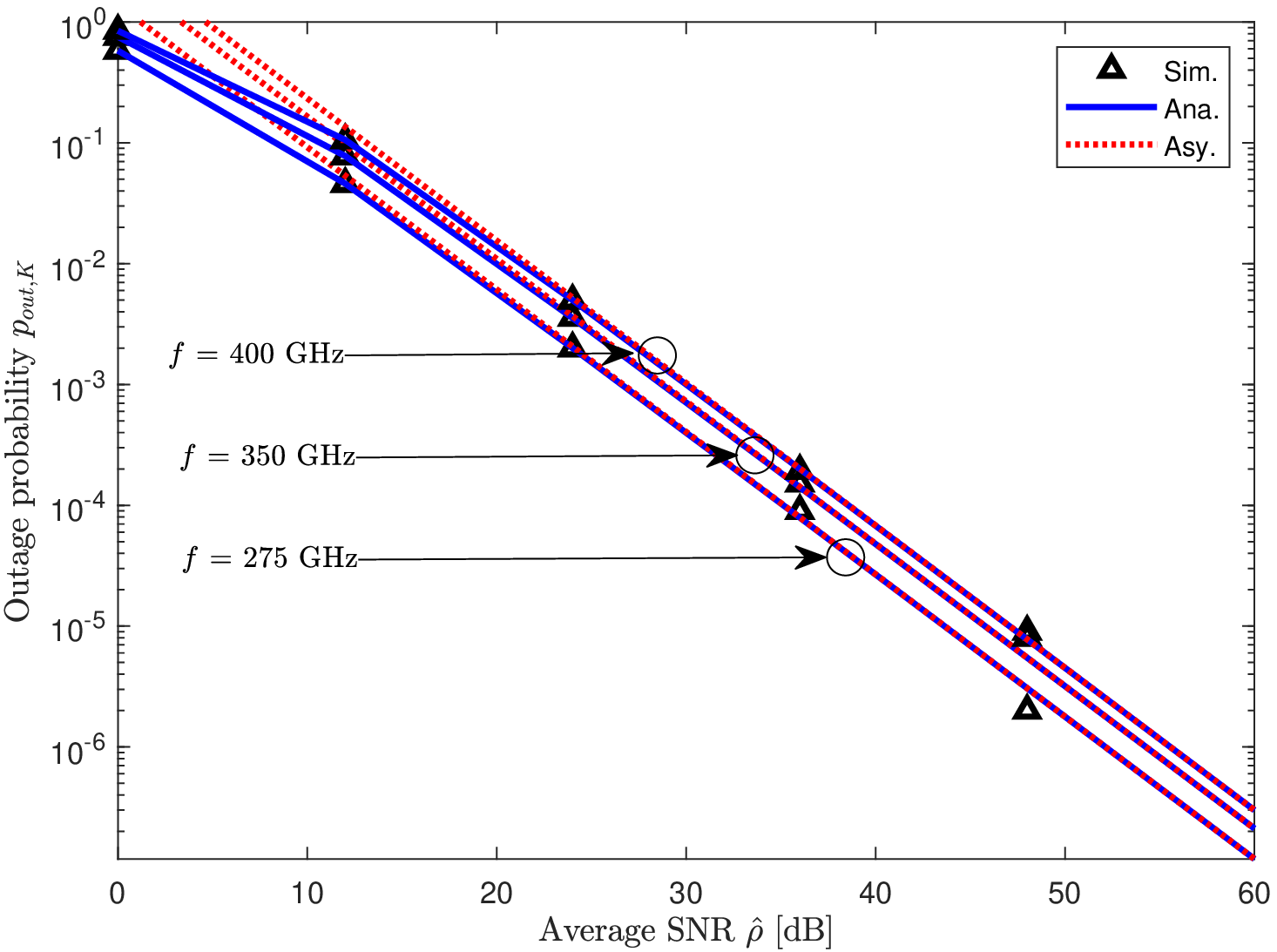}\\
        \caption{The outage probability $p_{out,K}$ versus the average SNR $\hat \rho $ for different carrier frequencies.}\label{fig:new1}
\end{figure}

\begin{figure}
        \centering
        \includegraphics[width=3.5in]{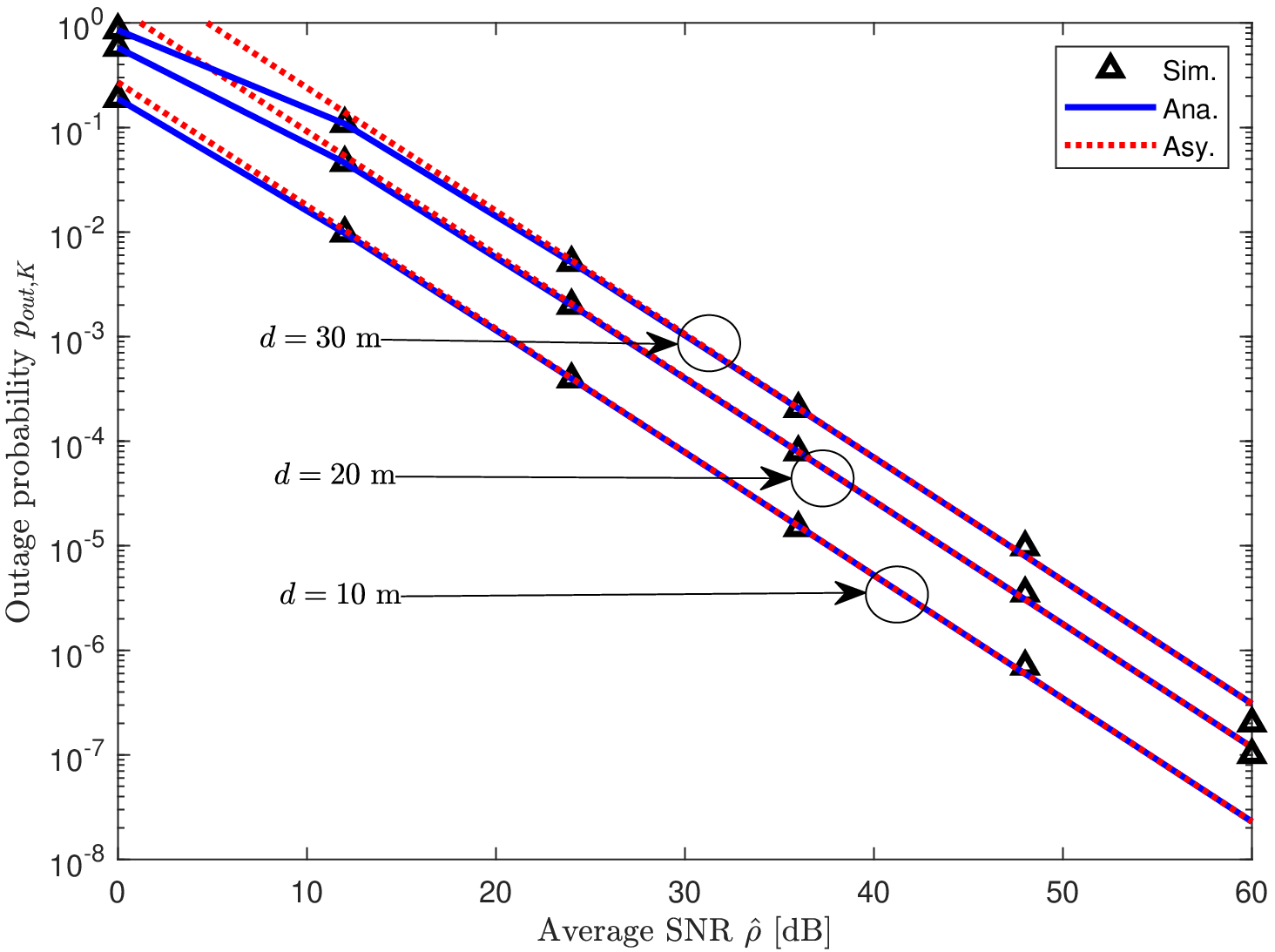}\\
        \caption{The outage probability $p_{out,K}$ versus the average SNR $\hat \rho $ for different transmission distances.}\label{fig:R12d}
\end{figure}


In Fig. \ref{fig:R2fso}, the outage performance of HARQ-IR aided THz communications is compared with that of HARQ-IR aided free space optical (FSO) communications, wherein the outage expression of HARQ-IR aided FSO communications was derived in \cite{8450318,jurado2011unifying,ansari2015performance}. For illustration, no independent scattering component and intensity modulation/direct detection (IM/DD) technique are assumed for FSO systems. It can be seen from Fig. \ref{fig:R2fso} that the HARQ-IR aided THz communications achieve a better outage performance than HARQ-IR aided FSO communications. This can be explained by the conclusion drawn in \cite{VarotsosJun2022} that the FSO system frequently suffers from severe atmospheric turbulence and pointing errors compared to the THz system. Hence, our results coincide with the experimental observations in \cite{VarotsosJun2022}. 
\begin{figure}
    \centering
    \includegraphics[width=3.5in]{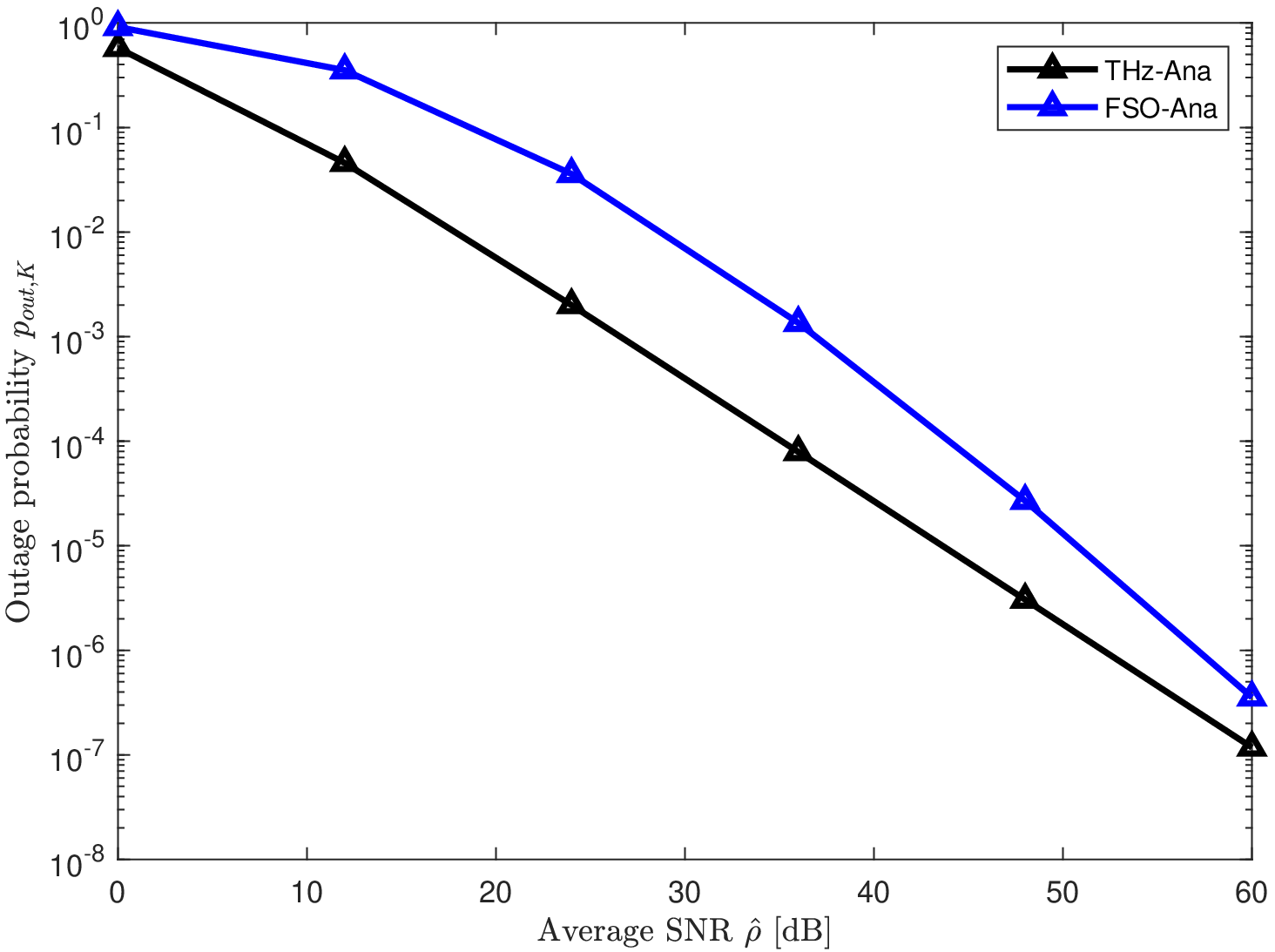}\\
    \caption{Comparison between the outage probabilities of HARQ-IR aided THz systems and FSO systems with $R=2$~bps/Hz and $K=3$.}\label{fig:R2fso}
\end{figure}

\subsection{With Blockage}
Figs. \ref{fig:R3} and \ref{fig:R4} show the outage probabilities ${p^{(L)}_{out,K}}$ versus the average transmit SNR $\hat \rho$ for single-hop and multi-hop HARQ-aided THz communication schemes in the presence of blockage, respectively, where the values of ${\lambda _b}$ are set as $0.03$ and $0.01$ for the single-hop and the multi-hop systems, respectively, $L$ is set as $2$. 
Clearly, the analytical and simulated results are in perfect agreement, which demonstrates the correctness of our theoretical analysis. Similarly, the HARQ-IR-aided scheme performs the best no matter under single-hop or multi-hop. In addition, it is not difficult to find that the outage probabilities of both single-hop and multi-hop systems have a lower bound. This is because the outage probability under high SNR is only related to the blocking probability, which coincides with Remark \ref{eqn:mh_do}. For instance, the lower bound of ${p^{(2)}_{out,K}}$ in Fig. \ref{fig:R4} is ${p^{(2)}_{out,K}} \ge 1-{P_N}^L = 1-e^{-2\beta d} = 0.007$. Hence, the outage probability eventually tends to the blocking probability according to \eqref{eqn:nbl}. Particularly, by comparing Figs. \ref{fig:R3} and \ref{fig:R4}, the outage probability of multi-hop systems is lower than that of single-hop systems, which corroborates the superiority of the multi-hop system. In addition, to verify the results in the case of $P_B = 0$ in Remark \ref{eqn:mh_do}, Fig. \ref{fig:Rt} is plotted to investigate the asymptotic outage behavior of multi-hop THz systems without blockage by setting ${\lambda _b}=0$. It is easily found that the outage probability decreases to zero as the SNR increases. The declining rate of the outage curves corroborates the validity of the diversity order analysis.


\begin{figure}
        \centering
        \includegraphics[width=3.5in]{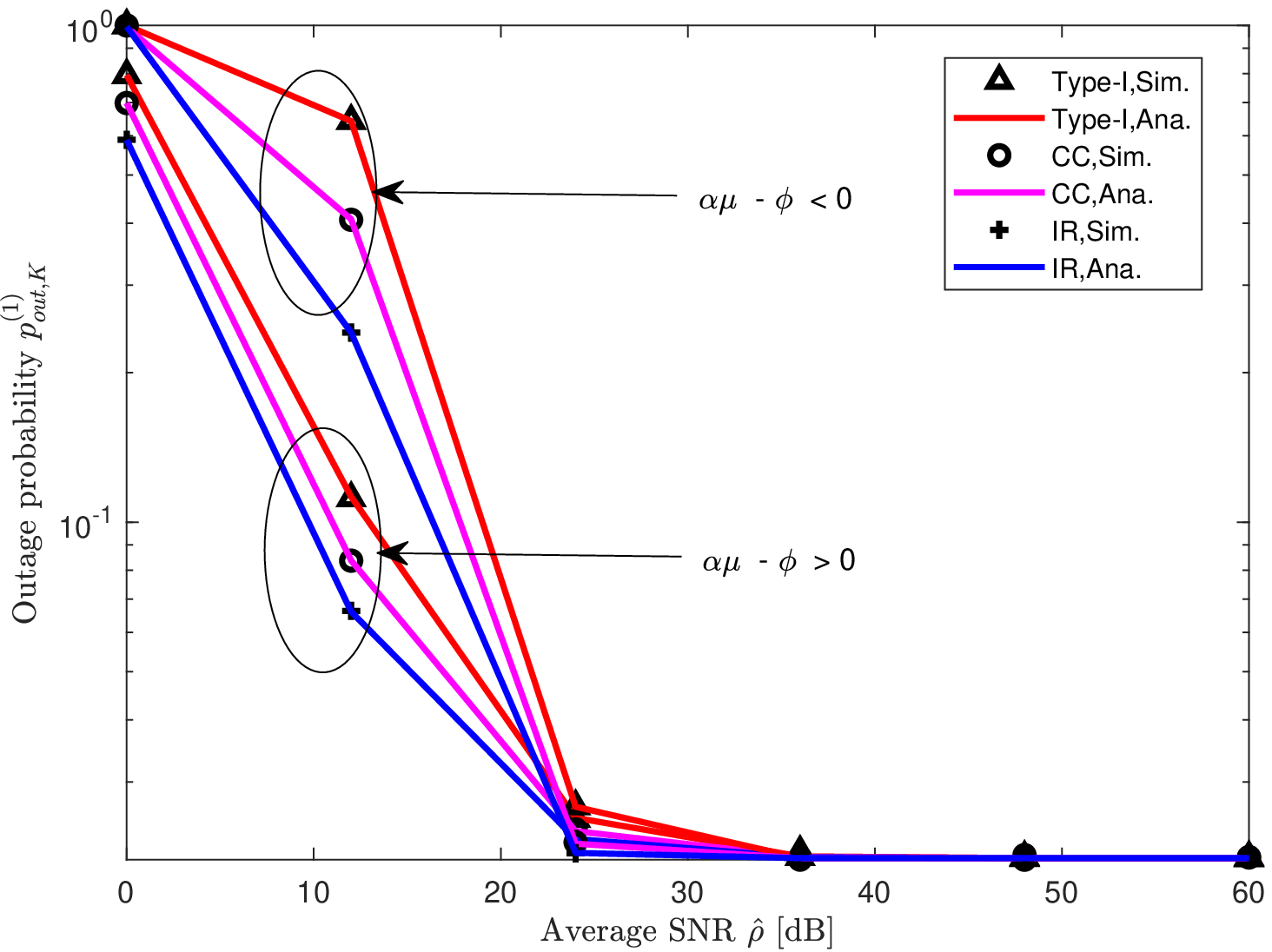}
        \caption{\textcolor[rgb]{0.00,0.00,0.00}{The outage probability $p_{out,K}^{\left( 1 \right)}$ versus the average SNR $\hat \rho $ with ${\lambda _b}=0.03$ in the case of single-hop}.}\label{fig:R3}
\end{figure}
\begin{figure}
        \centering
        \includegraphics[width=3.5in]{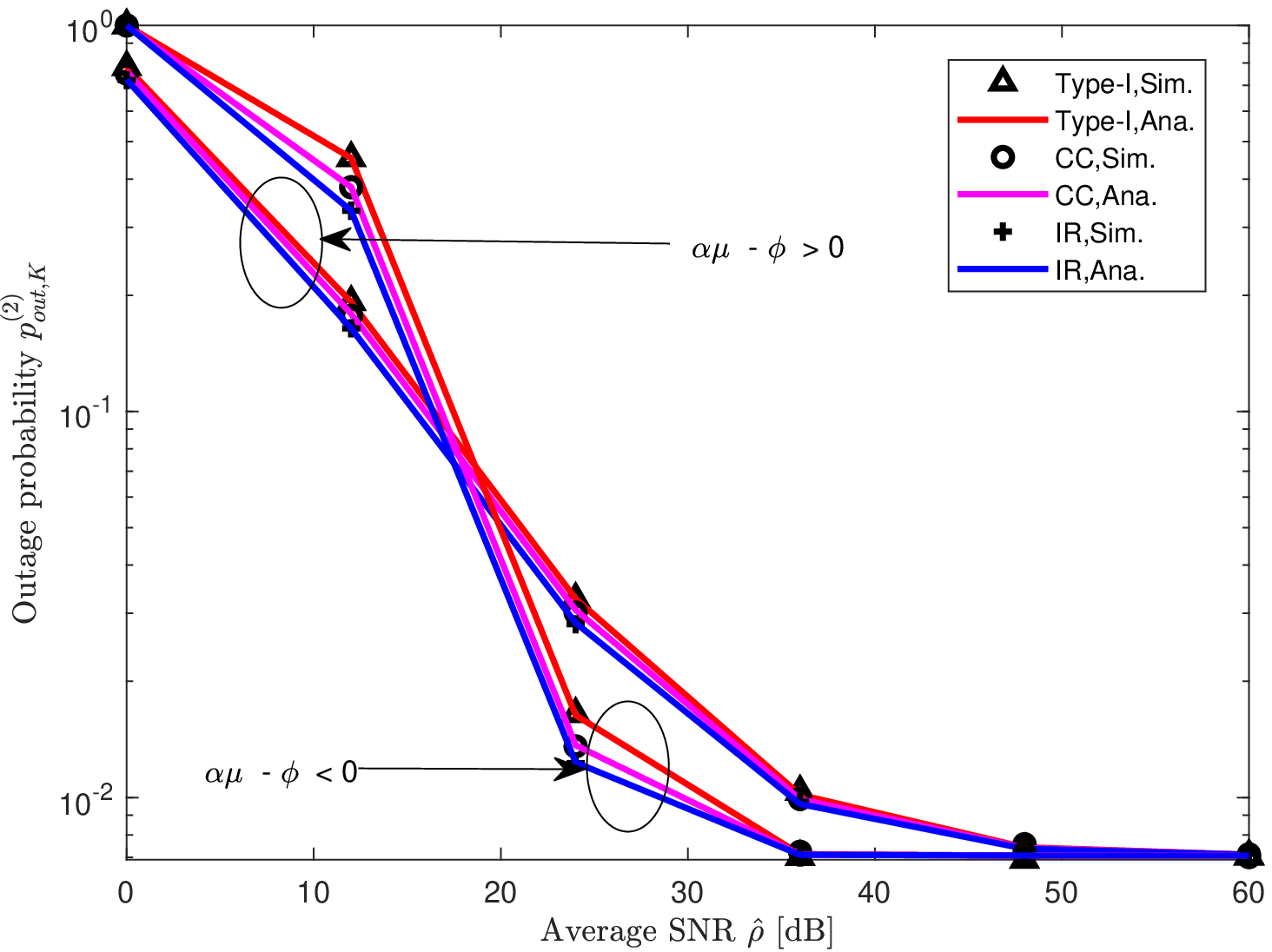}
        \caption{\textcolor[rgb]{0.00,0.00,0.00}{The outage probability $p_{out,K}^{\left( 2 \right)}$ versus the average SNR $\hat \rho $ with ${\lambda _b}=0.01$ and $L=2$  in the case of multi-hop}.}\label{fig:R4}
\end{figure}

\begin{figure}
        \centering
        \includegraphics[width=3.5in]{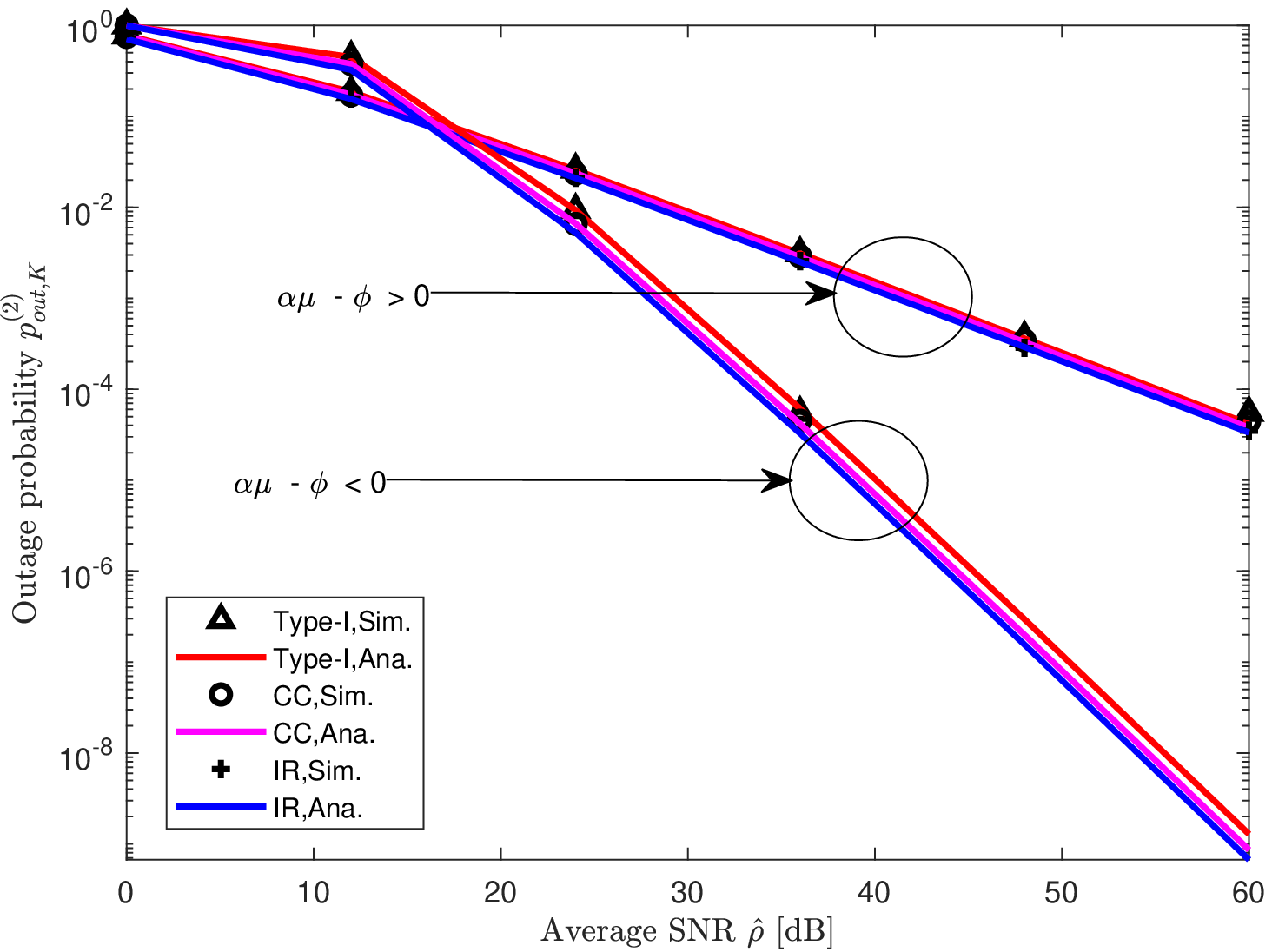}
        \caption{The outage probability $p_{out,K}^{\left( 2 \right)}$ versus the average SNR $\hat \rho $ with ${\lambda _b}=0$ and $L=2$ in the case of multi-hop}.\label{fig:Rt}
\end{figure}

As shown in Figs. \ref{fig:R5} and \ref{fig:R6}, the LTAT is plotted against the average SNR for single-hop and multi-hop HARQ-aided THz communication schemes in the presence of blockage, respectively. Clearly, both of them show an excellent match between the simulation results and the analytical results, which validate our analysis. Unsurprisingly, the LTAT can be improved through increasing the average SNR. It can be seen that HARQ-IR-aided scheme performs the best among the three HARQ-aided schemes, but the LTAT of all the HARQ-aided schemes tend to be fixed under high SNR, because the LTAT is upper bounded by ${{{\bar {\cal T}}}^{\left( L \right)}} \le R/L$ according to \eqref{eqn:LTAT_2} and \eqref{eqn:LTAT_5}. For example, the LTAT of the single-hop systems is bounded as ${{{\bar {\cal T}}}^{\left( L \right)}} \le R=2$~bps/Hz, while the LTAT of the single-hop systems is bounded as ${{{\bar {\cal T}}}^{\left( L \right)}} \le R/L=1$~bps/Hz. By combining the results in Figs. \ref{fig:R4} and \ref{fig:R6}, we conclude that the higher reliability of the multi-hop system is achieved at the price of lower spectral efficiency.

\begin{figure}
        \centering
        \includegraphics[width=3.5in]{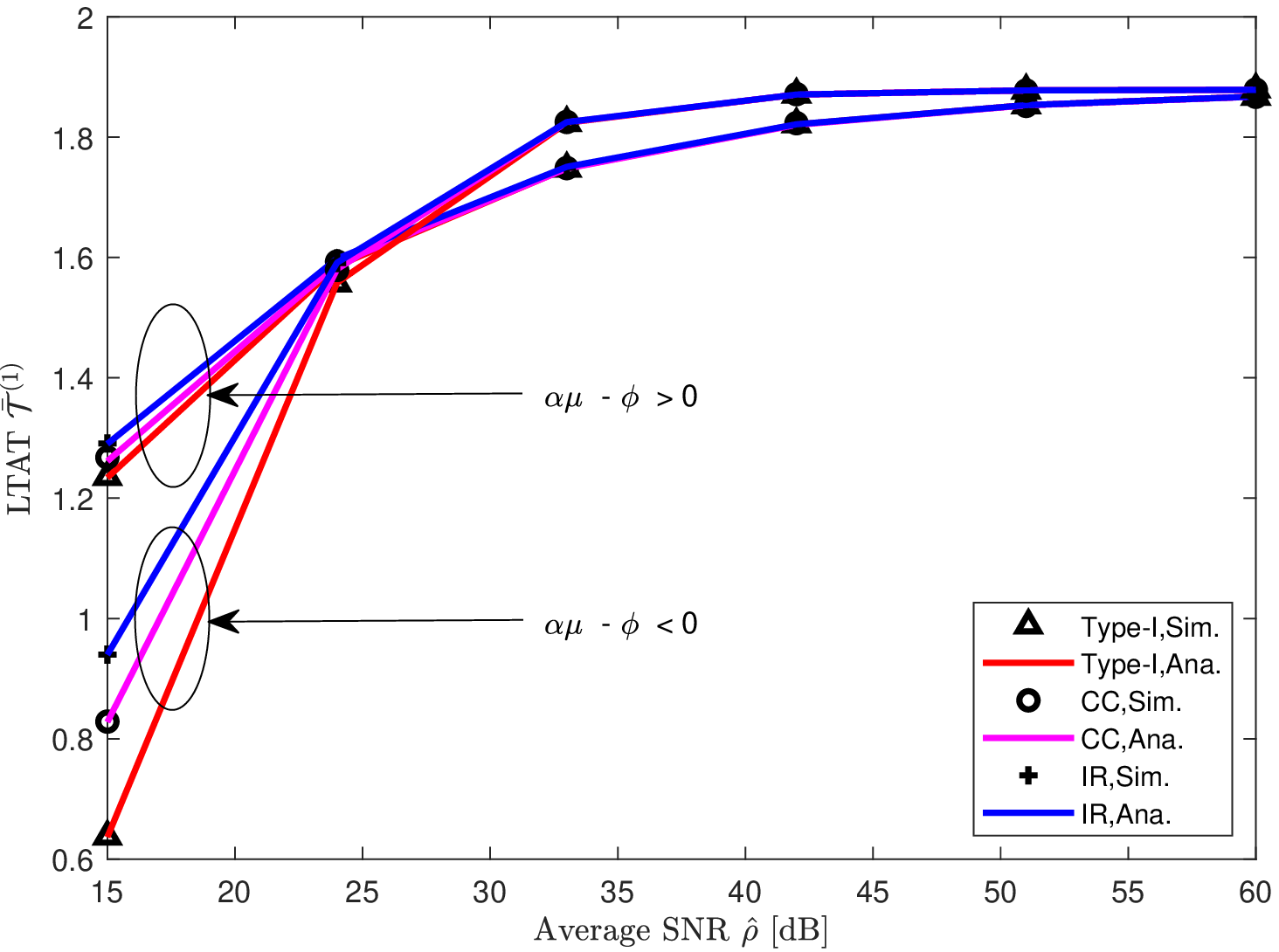}
        \caption{\textcolor[rgb]{0.00,0.00,0.00}{The LTAT ${\bar {\cal T}}^{\left( 1 \right)}$ versus the average SNR $\hat \rho $ with ${\lambda _b}=0.03$ in the case of single-hop}.}\label{fig:R5}
\end{figure}

\begin{figure}
        \centering
        \includegraphics[width=3.5in]{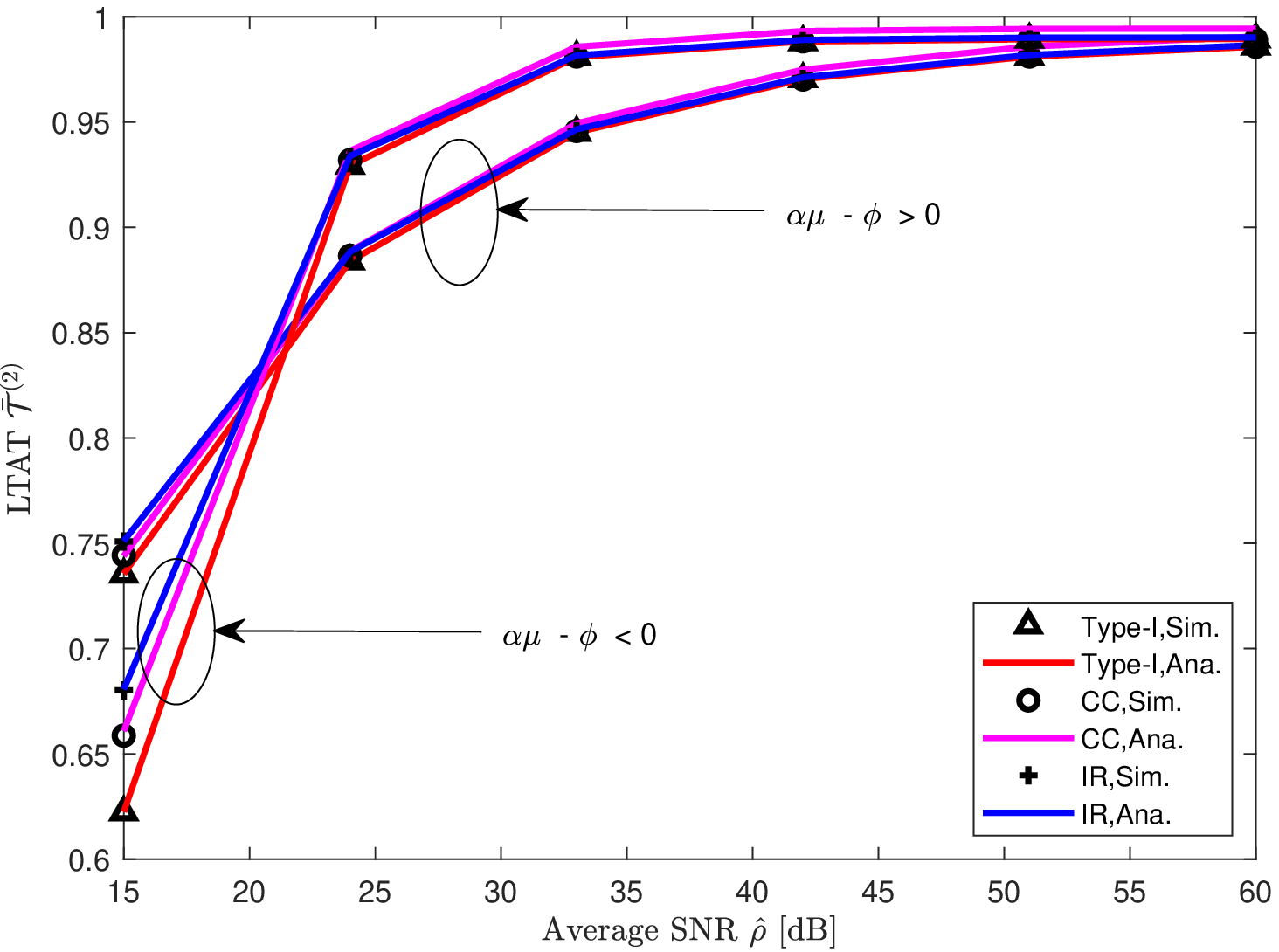}
        \caption{\textcolor[rgb]{0.00,0.00,0.00}{The LTAT ${\bar {\cal T}}^{\left( 2 \right)}$ versus the average SNR $\hat \rho $ with ${\lambda _b}=0.01$ and $L=2$ in the case of multi-hop}.}\label{fig:R6}
\end{figure}

\subsection{Deep Neural Network Based Outage Evaluation}
Since the analytical outage expression of HARQ-IR-aided THz communications involves the Fox's H function, it entails a prohibitively high computational complexity on the performance evaluation. Instead, we propose to utilize a deep neural network (DNN) to estimate the outage probability in a real-time fashion. As opposed to the prior works ({e.g., \cite{9581288}) that the exact results are commonly used to generate the output of the DNN for the dataset, the simulated and asymptotic results are leveraged in this paper to further shorten the offline execution time and lower the computational complexity. This is due to the fact that the asymptotic results and the simulated results can provide a fairly high computational accuracy in low and high outage regimes, respectively. In addition, it is worth mentioning that the computational complexity of the asymptotic outage probability is extremely low relative to the analytical calculation approach. Besides, there is a broad consensus that we can obtain a relatively high accuracy for evaluating a medium-to-high outage probability, even if a small number of simulation runs are carried out. This conclusion can be drawn from Table \ref{tab:EQ16}, wherein the simulation method has a lower execute time than the analytical method for evaluating $p_{out,K}>10^{-5}$, meanwhile guaranteeing a similar accuracy. This dramatic merit of Monte Carlo simulations is helpful to cut down the computation overhead. Thereby, the natural idea is to properly amalgamate the simulated and asymptotic results to generate the dataset. In order to elaborate on how the DNN is implemented to fit the outage curves, DNN structure,  dataset generation, DNN training, performance indicator, and numerical results are detailed as follows.

\begin{table*}
  \centering
  \caption{The computational complexity and accuracy comparisons between the simulation and the analytical methods with $K=3$ and $R=2$~bps/Hz.}
  \label{tab:EQ16}
        \begin{tabular}{|c||c|c|c|c|c|c|}
        \hline
        \multirow{2}*{SNR [db]}& Direct Integral & \multicolumn{3}{c|}{Simulation} & \multicolumn{2}{c|}{Analytical}\\
        \cline{2-7}
        &value&trials&value&time [s]&value&time [s]\\ 
        \hline
        \hline
        0&0.5792&${10^{4}}$&0.5816&0.147214&0.5797&181.328981\\
        \cline{1-7}
        12&0.0459&${10^{5}}$&0.0461&1.424897&0.0460&247.631133\\
        \cline{1-7}
        24&0.0020&${10^{6}}$&0.0020&12.247065&0.0020&278.964050\\
        \cline{1-7}
        36&7.9044e-05&${10^{8}}$&7.9920e-05&88.011415&7.9078e-05&275.858877\\
        \cline{1-7}
        48&3.0607e-06&${10^{9}}$&3.0710e-06&849.788840&3.0620e-06&302.424123\\
        \cline{1-7}
        60&1.1811e-07&${10^{10}}$&1.2380e-07&8450.979962&1.1816e-07&326.277437\\
        \hline
        \end{tabular}
\end{table*}



\subsubsection{DNN Structure}
A fully connected DNN-based outage evaluation framework is outlined in Fig. \ref{fig:R7}. The DNN consists of one input layer, two multiple hidden layers and one output layer. By taking $4$ system parameters as the inputs of DNN and the outage probability as the output, the input and the output layers have $4$ neurons and $1$ neuron, respectively. Herein, the four input parameters are the average transmit signal-to-noise ratio $\hat{\rho }$, the initial transmission rate $R$, the maximum number of transmissions $K$ and the radius of the receive antenna effective area ${w_{{d_0}}}$. The possible values of the system parameters are specified in Table \ref{table_2}. Besides, both the two hidden layers contain $100$ neurons. To avoid the vanishing gradient problem while ensuring the computational simplicity, all the neurons in the hidden and output layers adopt the exponential linear unit (ELU) activation function. 
\begin{table}[htbp]
  \centering
  \caption{Parameter values for DNN training and testing}\label{table_2},
  \begin{tabular}{|c|c|c|c|c|c|}
  \hline
  Parameters & Values & Parameters & Values & Parameters & Values \\
  \hline
  $\hat{\rho }$ & [0,50] & $\alpha$ & 2 & ${\sigma _{\rm{s}}}$ & 1 \\
  \hline
  $R$ & [0,5] & $\mu$ & 1 & $f$ & 275 \\
  \hline
  $K$ & [2,4] & $d$ & 20 & $T$ & 296 \\
  \hline
  ${w_{{d_0}}}$ & [3,4] & $\psi$ & 0.5 & $p$ & 101325 \\
  \hline
\end{tabular}
\end{table}

\begin{figure}
        \centering
        \includegraphics[width=3in]{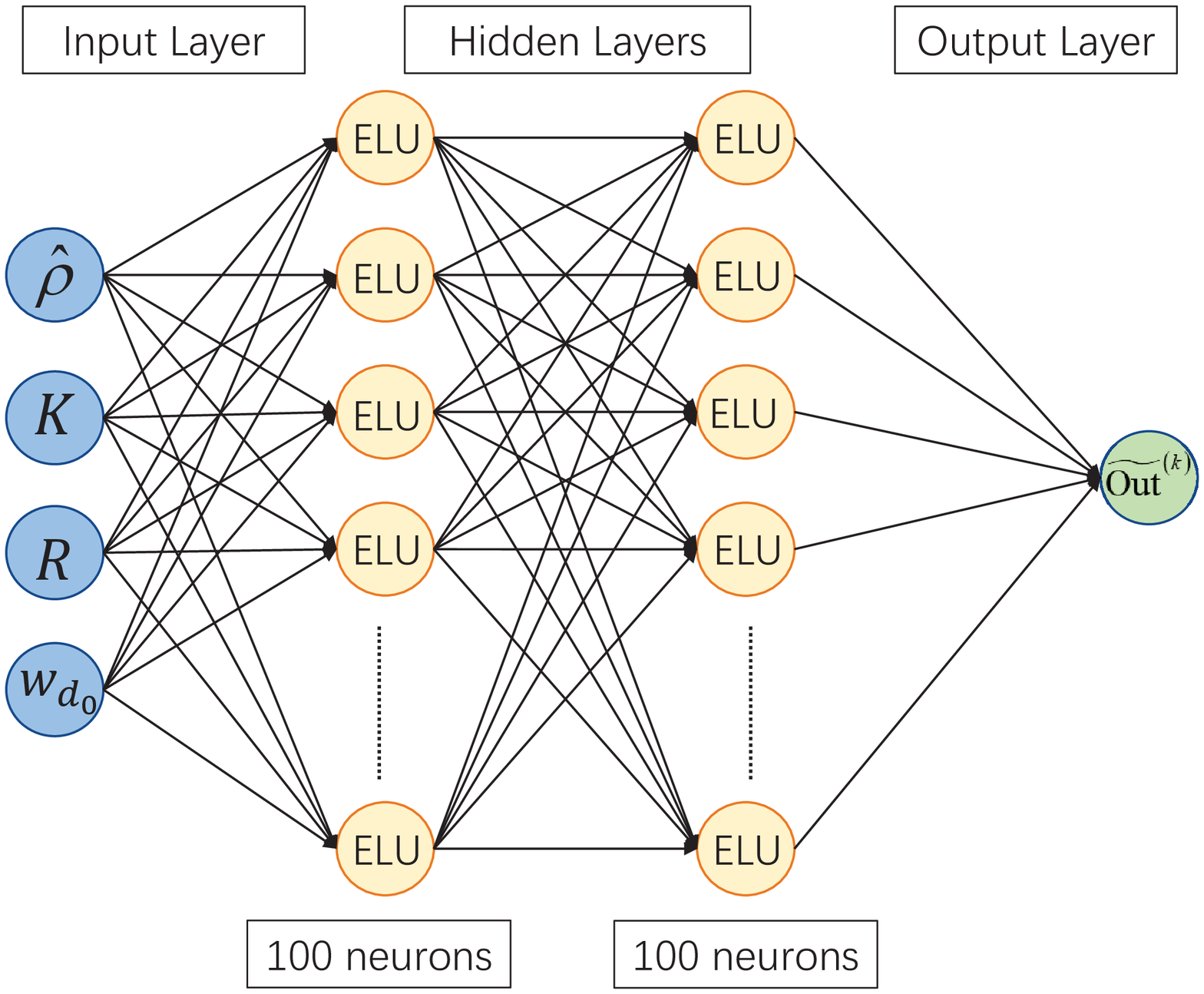}\\
        \caption{The framework of the proposed deep neural network based outage evaluation model.}\label{fig:R7}
\end{figure}


\subsubsection{Dataset Generation}
In the generated dataset ${{{\cal M}}}$, the $k$-th sample can be represented by $({\bf In}^{(k)},{\rm Out}^{(k)})$, where ${\bf In}^{(k)}$ contains four elements as provided in Table \ref{table_2}, ${\rm Out}^{(k)}$ corresponds to the output value which is obtained by combining the asymptotic and simulated results. The dataset ${{{\cal M}}}$ is generated in the following way. Regarding ${\bf In}^{(k)}$, the four inputs are uniformly generated according to the specified ranges in Table \ref{table_2}. According to the input parameters ${\bf In}^{(k)}$, the output ${\rm Out}^{(k)}$ is generated by combining the asymptotic and simulated outage probabilities. In order to reap the benefits of using both asymptotic and simulation results in low and high outage regimes, ${\rm Out}^{(k)}$ can be set as
\begin{equation}\label{eqn:outk}
{\rm{Out}}^{\left( k \right)} = \left\{ {\begin{array}{*{20}{c}}
{p_{out,K}^{sim}}&{p_{out,K}^{sim} \ge \upsilon }\\
{p_{out,K}^{asy}}&{\rm otherwise}
\end{array}} \right.,
\end{equation}
where $p_{out,K}^{sim}$ and ${p_{out,K}^{asy}}$ are the simulated and the asymptotic outage probabilities given the input parameters ${\bf In}^{(k)}$, $\upsilon$ is used to manage the computational accuracy. It should be highlighted that the improvement of the evaluation accuracy of the DNN-based framework is negligible under a sufficiently low $\upsilon$, because the asymptotic results are accurate enough in this case. This can be observed from Figs. \ref{fig:R1}, \ref{fig:new1}, and \ref{fig:R12d}  that the simulated results coincide well with the asymptotic ones. It is noteworthy that we are more interested in the low outage region for reliable communications. Unfortunately, numerous experiments indicate that \eqref{eqn:outk} is not an appropriate choice for defining the output, because the setting of ${\rm Out}^{(k)}$ in \eqref{eqn:outk} makes the impact of low outage region on the loss function nearly negligible. Hence, to better predict the low outage probability and meanwhile accelerate learning and alleviate the over-fitting, the following transformation is utilized such that the low outage is magnified meanwhile the accuracy of the prediction of high outage is also warranted
\begin{equation}\label{eqn:out_trans}
\widetilde {\rm{Out}}^{\left( k \right)} = {\log _2}\left( { - {{\log }_2}\left({\rm{Out}}^{\left( k \right)}\right)} \right).
\end{equation}

The generated dataset ${{{\cal M}}}$ is then partitioned into three non-overlapping subsets, including the training set ${{{\cal M}}_{{\rm{tra}}}}$, the validation set ${{{\cal M}}_{{\rm{val}}}}$, and the testing set ${{{\cal M}}_{{\rm{tes}}}}$. Particularly, ${{{\cal M}}_{{\rm{tra}}}}$ is used to iteratively update the weights and biases of the DNN, ${{{\cal M}}_{{\rm{val}}}}$ is used to terminate the training phase and avoid the occurrence of the over-fitting, ${{{\cal M}}_{{\rm{tes}}}}$ is used to assess the performance of the designed DNN for prediction. Moreover, the data in ${{{\cal M}}_{{\rm{tra}}}}$, ${{{\cal M}}_{{\rm{val}}}}$, and ${{{\cal M}}_{{\rm{tes}}}}$ account for $60\% $, $20\% $, and $20\% $ of the whole dataset ${{{\cal M}}}$, respectively.

\subsubsection{DNN Training}
The deep learning (DL) procedure is split into two stages, including the training stage and the testing stage. Specifically, in the training stage, the adaptive moment estimation (Adam) optimization algorithm is adopted to optimize the weights and biases of the DNN offline. Moreover, the loss function is defined as the mean squared error (MSE) between the actual and the predicted output values. 
Once the training stage is finalized, the obtained DNN can be used for the online prediction. 

\subsubsection{Performance Indicator}
To assess the performance of the designed DNN-based outage evaluation model, the MSE is utilized as the performance indicator. Regarding the testing set, the MSE $\mathcal E_{\rm mse}$ can be calculated as
\begin{align}\label{eqn:PE}
\mathcal E_{\rm mse} = \frac{1}{|{{{\cal M}}_{{\rm{tes}}}}|}\sum\limits_{k = 1}^{|{{{\cal M}}_{{\rm{tes}}}}|} {{{\left( {\widetilde {\rm{Out}}^{\left( k \right)} - \widetilde {\rm{Out}}_{\rm pre}^{\left( k \right)}} \right)}^2}},
\end{align}
where $|{{{\cal M}}_{{\rm{tes}}}}|$ denotes the cardinality of the testing set, $\widetilde {\rm{Out}}_{\rm pre}^{\left( k \right)}$ refers to the predicted value of the DNN. 

\subsubsection{Numerical Results}
For verification, Matlab is used to collect the dataset $\mathcal M$ as well as analyze the data, and the DNN-based outage evaluation model is implemented in a PC with an AMD Ryzen 7 4800H, 64-core processor, Nvidia GeForce RTX-2060 super GPUs, Python 3.7.9 and Pytorch 1.11.0 as the DL framework. As shown in Fig. \ref{fig:R8}, the MSE $\mathcal E_{\rm mse}$ is plotted against the total number of samples in the dataset $\mathcal M$. It can be seen that the MSE gradually decreases with the increase of the total number of samples. Nevertheless, once the total number of samples is larger than $ 1.25 \times {10^4}$, the MSE of the trained DNN-based model closely approaches to the performance floor, i.e., $\mathcal E_{\rm mse} \approx 0.03$. Thus the number of samples should be properly chosen to balance the tradeoff between the accuracy and complexity. Moreover, it is observed that the MSE values of the training, the validation and the testing sets are close to each other. This finding indicates that the DNN is well trained \cite{8738823}.
\begin{figure}
        \centering
        \includegraphics[width=3.5in]{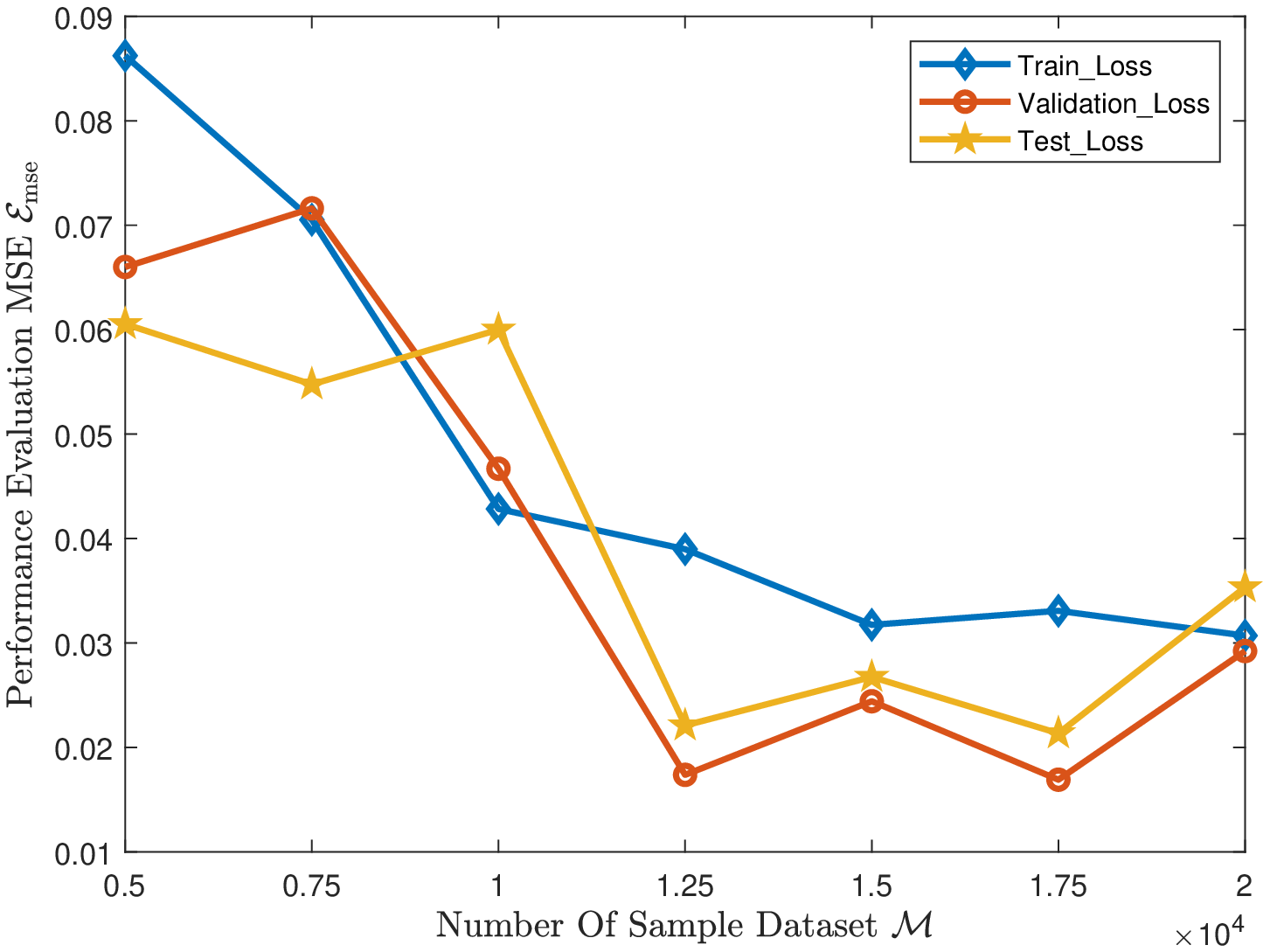}\\
        \caption{The convergence curve of the MSE of the proposed DNN-based outage evaluation model with $\upsilon  = {10^{ - 4}}$.}\label{fig:R8}
\end{figure}
\subsection{Optimal Rate Selection}
To illustrate the benefits of using our analytical results, the transmission rate is optimally designed to maximize the LTAT while guaranteeing the outage constraint in this section. From \eqref{eqn:LTAT_1}, it is found that the LTAT is a complex function of the transmission rate, because both the numerator and denominator involve the transmission rate. Hence, it is very cumbersome to apply the analytical result to solve the optimization problem. Thereupon, another alternative way is using the asymptotic result to reduce the computational complexity. Nonetheless, this approach may yield a nonnegligible gap under low SNR by comparing to the true optimal transmission rate. Besides, this motivates us to utilize the DNN-based outage evaluation method to simplify the optimization meanwhile ensuring the computational accuracy.
By taking the case without blockage as an example, the optimal rate design problem can be casted as
\begin{equation}\label{eqn:rate}
        \begin{array}{*{20}{c}}
          {\mathop {\max }\limits_R }&{\overline {{\cal T}} }\\
          {{\rm{s}}.{\rm{t}}.}&{{p_{out,K}} \le \varepsilon }
        \end{array}
\end{equation}
where $\varepsilon $ denotes the maximum allowable outage probability. It is noteworthy that the asymptotic result can be used to obtain a globally sub-optimal solution \cite{7959548}. This is because the numerator and denominator of $\bar {\mathcal T}$ are concave and convex functions of $R$, respectively, by considering the convexity of ${{{\cal G}_K}\left( {{2^R}} \right)}$, as proved in Section \ref{sec:cr}. Hence, \eqref{eqn:rate} is a concave fractional programming problem, which can be globally solved with Dinkelbach's method \cite{dinkelbach1967nonlinear}. Additionally, the proposed DNN-based method can be used to accurately approximate the outage probability, which can be adopted to numerically solved this optimization problem. 


\begin{figure}
        \centering
        \includegraphics[width=3.5in]{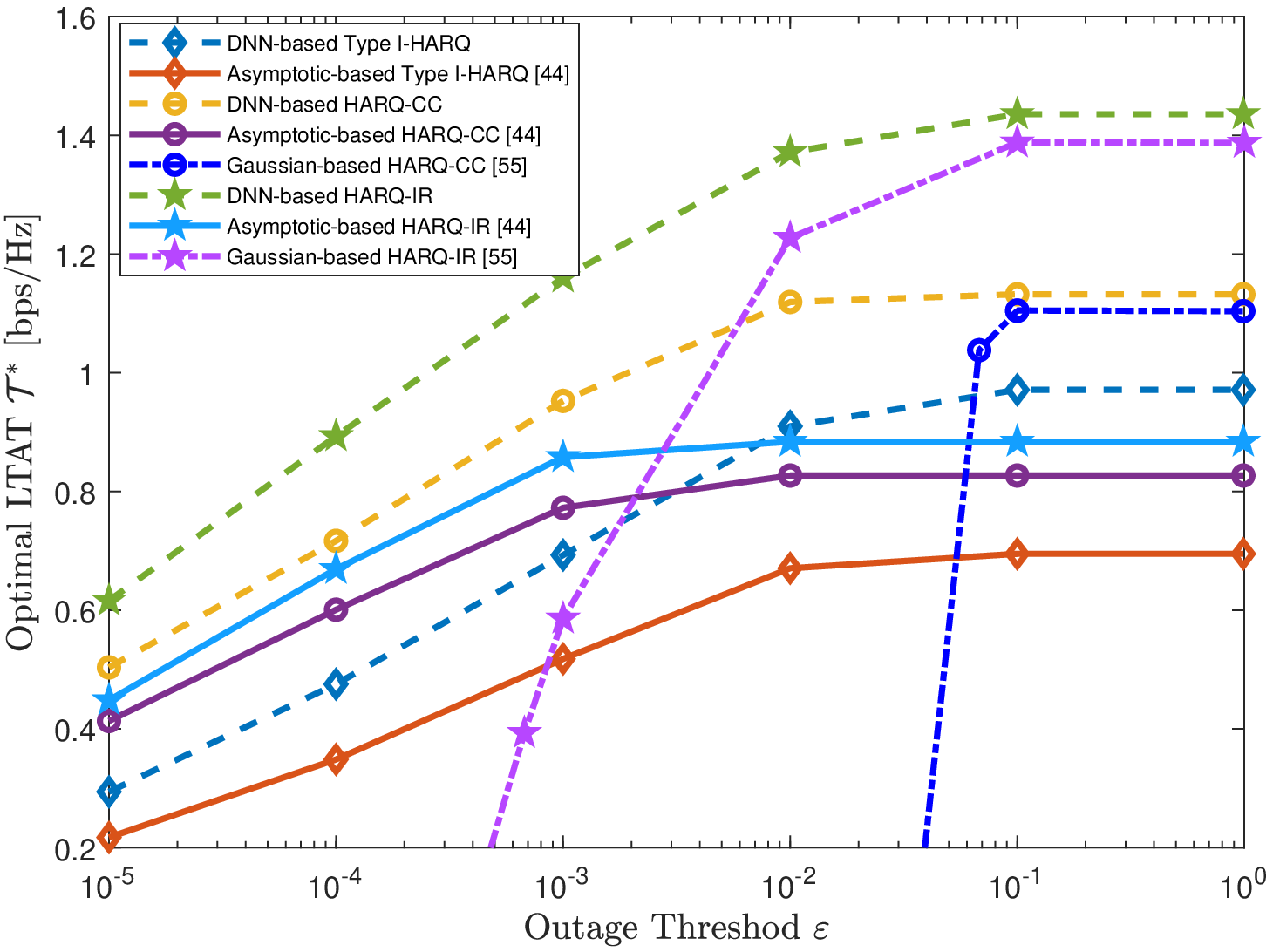}\\
        \caption{The optimal LTAT $\mathcal T^*$ versus the outage threshold $\varepsilon$ with $K=4$ and $\hat \rho=20$~dB.}\label{fig:R9}
\end{figure}
Fig. \ref{fig:R9} shows the comparison between the proposed DNN-based and the asymptotic-based methods in terms of the optimal LTAT for different HARQ-aided schemes, where the asymptotic-based method \cite{7959548} and the Gaussian-based method \cite{wu2010performance} are used for benchmarking purpose. As expected, the optimal LTAT $\mathcal T^*$ increases to an upper bound as the outage threshold increases. Similarly, it is found that the HARQ-IR-aided scheme performs the best. Moreover, it is observed that the DNN-based method outperforms the asymptotic-based and the Gaussian-based ones, which highlight the considerable significance of our contribution.




\section{Conclusions}\label{sec:con}
This paper has first investigated the outage performance of HARQ-IR-aided THz communications in the absence of the blockage. Particularly, the Mellin transform has been adopted to derive the outage probability, which has facilitated the asymptotic outage analysis to gain useful insights. In order to explore more insights, the asymptotic performance metrics of HARQ-IR-aided schemes have been investigated, including the pointing error and fading channels, power allocation, modulation and coding gain, and diversity order. 
Moreover, the LTAT of the HARQ-IR-aided scheme has been derived in terms of the outage probability. Furthermore, we have studied the outage and LTAT performance of HARQ-IR-aided THz communications in the presence of blockage, where a 3D blockage modeling has been considered. All the analytical results have been validated by the numerical analysis. 


Furthermore, to render the tractability of the analysis, the effect of line-of-sight (LoS) has been overlooked in this paper. However, the $\alpha$-$\mu$ fading is a generalized fading channel model that can closely approximate Rician fading to account for the impact of LoS. Nevertheless, the effect of LoS on the performance of THz communications will be accurately analyzed in our future work.


Last but not the least, a novel DNN-based outage evaluation framework has been designed to compute the outage probability of the HARQ-IR-aided THz systems. The core of the proposed framework lies in the fact that Monte Carlo simulations and asymptotic outage analysis can offer highly accurate outage evaluations in low SNR and high SNR, respectively. By combining the asymptotic and simulation results, the proposed framework has been capable of striking a balanced tradeoff between high accuracy and low complexity. In the end, this novel outage evaluation framework has been used for the optimal rate selection, which is superior to the asymptotic based optimization.



\appendices

\section{Proof of Lemma \ref{eqn:le_1}}\label{sec:le1}
By applying Parseval's Type Property of Mellin transform \cite[Eq. 8.3.22]{debnath2014integral} to the inner integral of (\ref{eqn:IR_4}), one has
\begin{align}\label{eqn:IR_5}
p_{out,K} =& \frac{1}{{2\pi \rm i}}\int_{{{\rm{c}}_1} - \rm i\infty }^{{{\rm{c}}_1}{\rm{ + i}}\infty } {\frac{{{2^{ - Rt}}}}{{ - t}}\prod\limits_{k = 1}^K {\frac{\phi }{{2\Gamma (\mu )\Gamma \left( { - t} \right)}}}}\notag\\
&\times\frac{1}{{2\pi \rm i}}\int_{{{\rm{c}}_2} - \rm i\infty }^{{{\rm{c}}_2}{\rm{ + i}}\infty } {\frac{{\Gamma \left( {\frac{s}{2}} \right)\Gamma \left( { - t - \frac{s}{2}} \right)\Gamma (\mu  - \frac{s}{\alpha })\Gamma \left( {\phi  - s} \right)}}{{\Gamma \left( {1 + \phi  - s} \right)}}}\notag\\
&\times{{\left( {{{\left( {{\rho _k} {{\left| {{h_l}} \right|}^2}} \right)}^{ - \frac{1}{2}}}{{\left( {\frac{\mu }{{\hat h_f^\alpha S_{\rm{0}}^\alpha }}} \right)}^{\frac{1}{\alpha }}}} \right)}^s}ds dt,
\end{align}
where ${{{\rm{c}}_2}}  > 0$ and (\ref{eqn:IR_5}) holds by using the following two Mellin transforms
\begin{multline}\label{eqn:IR_6}
\int_0^\infty  {x_k^{s - 1}\Gamma \left( {\frac{{\alpha \mu  - \phi }}{\alpha },\frac{{\mu x_k^\alpha }}{{S_0^\alpha \hat h_f^\alpha }}} \right)} d{x_k}\\
= \frac{1}{s}{\left( {\frac{\mu }{{S_0^\alpha \hat h_f^\alpha }}} \right)^{ - \frac{s}{\alpha }}}\Gamma \left( {\frac{{\alpha \mu  - \phi  + s}}{\alpha }} \right),
\end{multline}
\begin{multline}\label{eqn:IR_7}
\int_0^\infty  {x_k^{s - 1}{{(1 + {\rho _k} {{\left| {{h_l}} \right|}^2}x_k^2)}^t}} d{x_k} \\
= \frac{1}{2}{\left( {\frac{1}{{\rho _k {{\left| {{h_l}} \right|}^2}}}} \right)^{\frac{s}{2}}}\frac{{\Gamma \left( {\frac{s}{2}} \right)\Gamma \left( { - t - \frac{s}{2}} \right)}}{{\Gamma \left( { - t} \right)}},
\end{multline}
where \eqref{eqn:IR_6} and \eqref{eqn:IR_7} follows from \cite[Eq. 6.455]{gradshteyn2014table} and \cite[Eq. A-1.24]{debnath2014integral}, respectively. By recognizing the inner integral as a Mellin-Barnes integral and identifying it with the Fox's $\text H$-function \cite[Eq.1.2]{mathai2009h}, the outage probability of HARQ-IR-aided THz communications is consequently derived as (\ref{eqn:IR_8}), as shown at the top of the next page. Obviously, (\ref{eqn:IR_8}) can be expressed in the form of the inverse Laplace transform by a variable substitution as (\ref{eqn:IR_9}).
\begin{figure*}[!t]
\begin{equation}\label{eqn:IR_8}
p_{out,K} = {\left( {\frac{\phi }{{2\Gamma (\mu )}}} \right)^K}\frac{1}{{2\pi {\rm{i}}}}\int_{{{\rm{c}}_1} - \rm i\infty }^{{{\rm{c}}_1}{\rm{ + i}}\infty } {\frac{{{2^{ - Rt}}}}{{ - t}}\prod\limits_{k = 1}^K {\frac{1}{{\Gamma \left( { - t} \right)}}H_{3,2}^{1,3}\left[ {{{\left( {\rho _k {{\left| {{h_l}} \right|}^2}} \right)}^{\frac{1}{2}}}{{\left( {\frac{\mu }{{\hat h_f^\alpha S_{\rm{0}}^\alpha }}} \right)}^{ - \frac{1}{\alpha }}}\left| {_{\left( {0,\frac{1}{2}} \right),\left( { - \phi ,1} \right)}^{\left( {1 + t,\frac{1}{2}} \right),\left( {1 - \mu ,\frac{1}{\alpha }} \right),\left( {1 + \phi ,1} \right)}} \right.} \right]} dt}
\end{equation}
\hrulefill
\end{figure*}

\section{Proof of Theorem \ref{eqn:le_2}}\label{sec:le2}
By using the residue theorem, the Fox's $\text H$-function in \eqref{eqn:IR_9} can be expanded as \eqref{eqn:res}, as shown at the top of the next page, where ${\mathop {{\rm{Res}}}\limits_{s = a} }\{f(s)\}$ denotes the residue of $f(s)$ at $s=a$.
\begin{figure*}[!t]
\begin{align}\label{eqn:res}
&{H_{3,2}^{1,3}\left[ {{{\left( {{q_k}\hat \rho {{\left| {{h_l}} \right|}^2}} \right)}^{\frac{1}{2}}}{{\left( {\frac{\mu }{{\hat h_f^\alpha S_{\rm{0}}^\alpha }}} \right)}^{ - \frac{1}{\alpha }}}\left| {_{\left( {0,\frac{1}{2}} \right),\left( { - \phi ,1} \right)}^{\left( {1 + t,\frac{1}{2}} \right),\left( {1 - \mu ,\frac{1}{\alpha }} \right),\left( {1 + \phi ,1} \right)}} \right.} \right]}\notag\\
 =&  \sum\nolimits_{\scriptstyle\left\{ {a:a = \phi  + n,n \in \left[ {0,\infty } \right]} \right\}\hfill\atop
{\scriptstyle\bigcup {\left\{ {a:a = \mu \alpha  + n\alpha ,n \in \left[ {0,\infty } \right]} \right\}} \hfill\atop
\scriptstyle\bigcup {\left\{ {a:a = 2t + 2n,n \in \left[ {0,\infty } \right]} \right\}} \hfill}} {\mathop {{\rm{Re}}s}\limits_{s = a} \left\{ {\frac{{\Gamma \left( {\frac{s}{2}} \right)\Gamma \left( { - t - \frac{s}{2}} \right)\Gamma (\mu  - \frac{s}{\alpha })\Gamma \left( {\phi  - s} \right)}}{{\Gamma \left( {1 + \phi  - s} \right)}}{{\left( {{{\left( {{{q_k}\hat \rho}{{\left| {{h_l}} \right|}^2}} \right)}^{ - \frac{1}{2}}}{{\left( {\frac{\mu }{{\hat h_f^\alpha S_{\rm{0}}^\alpha }}} \right)}^{\frac{1}{\alpha }}}} \right)}^s}} \right\}}
\end{align}
\hrulefill
\end{figure*}
Furthermore, by using the dominant term approximation together with the condition $t <  - \max \{ \phi ,\mu \alpha \} /2$, \eqref{eqn:res} is asymptotic to
\begin{align}\label{eqn:res1}
&{H_{3,2}^{1,3}\left[ {{{\left( {{q_k}\hat \rho {{\left| {{h_l}} \right|}^2}} \right)}^{\frac{1}{2}}}{{\left( {\frac{\mu }{{\hat h_f^\alpha S_{\rm{0}}^\alpha }}} \right)}^{ - \frac{1}{\alpha }}}\left| {_{\left( {0,\frac{1}{2}} \right),\left( { - \phi ,1} \right)}^{\left( {1 + t,\frac{1}{2}} \right),\left( {1 - \mu ,\frac{1}{\alpha }} \right),\left( {1 + \phi ,1} \right)}} \right.} \right]}\notag\\
\simeq & - \mathop {\lim }\limits_{s \to \phi } (s - \phi )\frac{{\Gamma \left( {\frac{s}{2}} \right)\Gamma \left( { - t - \frac{s}{2}} \right)\Gamma (\mu  - \frac{s}{\alpha })\Gamma \left( {\phi  - s} \right)}}{{\Gamma \left( {1 + \phi  - s} \right)}} \notag\\
&\times {\left( {{{\left( {{{q_k}\hat \rho} {{\left| {{h_l}} \right|}^2}} \right)}^{ - \frac{1}{2}}}{{\left( {\frac{\mu }{{\hat h_f^\alpha S_{\rm{0}}^\alpha }}} \right)}^{\frac{1}{\alpha }}}} \right)^s}\notag\\
 &- \mathop {\lim }\limits_{s \to \mu \alpha } (s - \mu \alpha )\frac{{\Gamma \left( {\frac{s}{2}} \right)\Gamma \left( { - t - \frac{s}{2}} \right)\Gamma (\mu  - \frac{s}{\alpha })\Gamma \left( {\phi  - s} \right)}}{{\Gamma \left( {1 + \phi  - s} \right)}}\notag\\
 &\times{\left( {{{\left( {{{q_k}\hat \rho} {{\left| {{h_l}} \right|}^2}} \right)}^{ - \frac{1}{2}}}{{\left( {\frac{\mu }{{\hat h_f^\alpha S_{\rm{0}}^\alpha }}} \right)}^{\frac{1}{\alpha }}}} \right)^s}\notag\\
  =& \left({{{q_k}\hat \rho}}\right)^{ - \frac{\phi }{2}}\Gamma \left( { - t - \frac{\phi }{2}} \right) B \notag\\
 &+ \left({{{q_k}\hat \rho}}\right)^{ - \frac{{\mu \alpha }}{2}}\Gamma \left( { - t - \frac{{\mu \alpha }}{2}} \right)C,
\end{align}
where $B$ and $C$ are given by
\begin{equation}\label{eqn:B}
B = \Gamma \left( {\frac{\phi }{2}} \right)\Gamma (\mu  - \frac{\phi }{\alpha }){{\left( {{{\left| {{h_l}} \right|}^{ - 1}}{{\left( {\frac{\mu }{{\hat h_f^\alpha S_{\rm{0}}^\alpha }}} \right)}^{\frac{1}{\alpha }}}} \right)}^\phi },
\end{equation}
\begin{equation}\label{eqn:C}
C = \frac{\alpha }{{\phi  - \mu \alpha }}\Gamma \left( {\frac{{\mu \alpha }}{2}} \right){{\left( {{{\left| {{h_l}} \right|}^{ - 1}}{{\left( {\frac{\mu }{{\hat h_f^\alpha S_{\rm{0}}^\alpha }}} \right)}^{\frac{1}{\alpha }}}} \right)}^{\mu \alpha }}.
\end{equation}

By substituting \eqref{eqn:res1} into \eqref{eqn:IR_9}, the asymptotic outage probability can be obtained as
\begin{align}\label{eqn:IR_90}
{p_{out,K}} &\simeq {\left( {\frac{\phi }{{2\Gamma (\mu )}}} \right)^K}\frac{1}{{2\pi {\text{i}}}}\int_{{{\text{c}}_1}-{\text{i}}\infty }^{{{\text{c}}_1}+{\text{i}}\infty } {\frac{{{2^{ - Rt}}}}{{ - t}}\prod\limits_{k = 1}^K }\notag\\
&\times{\frac{1}{{\Gamma \left( { - t} \right)}}\left( \begin{gathered}
  \left({{{q_k}\hat \rho}}\right)^{ - \frac{\phi }{2}}\Gamma \left( { - t - \frac{\phi }{2}} \right)B \hfill \\
   + \left({{{q_k}\hat \rho}}\right)^{ - \frac{{\mu \alpha }}{2}}\Gamma \left( { - t - \frac{{\mu \alpha }}{2}} \right)C \hfill \\
\end{gathered}  \right)} dt??
\end{align}
Furthermore, \eqref{eqn:IR_90} can be expanded as \eqref{eqn:IR_900}, as shown at the top of the next page,
\begin{figure*}[!t]
\begin{align}\label{eqn:IR_900}
p_{out,K} &\simeq {\left( {\frac{\phi }{{2\Gamma (\mu )}}} \right)^K}\frac{1}{{2\pi \rm i}}\int_{{{\rm{c}}_1} - \rm i\infty }^{{{\rm{c}}_1}{\rm{ + i}}\infty } {\frac{{{2^{ - Rt}}}}{{ - t}}}\prod\limits_{k = 1}^K {\left( {{{q_k^{ - \frac{\phi }{2}}{\hat \rho ^{ - \frac{\phi }{2}}}}}{B}\frac{{\Gamma \left( { - t - \frac{\phi }{2}} \right)}}{{\Gamma \left( { - t} \right)}} + {{q_k^{ - \frac{{\mu \alpha } }{2}}{\hat \rho ^{ - \frac{{\mu \alpha }}{2}}}}}{C}\frac{{\Gamma \left( { - t - \frac{{\mu \alpha }}{2}} \right)}}{{\Gamma \left( { - t} \right)}}} \right)} dt \notag\\
& = {\left( {\frac{\phi }{{2\Gamma (\mu )}}} \right)^K}\sum\limits_{{\bf{a}} \in {\bf{\Omega }}} {{{\hat \rho }^{ {\frac{{\mu \alpha  - \phi }}{2}} \sum\limits_{k = 1}^K {{a_k}}  - \frac{{K\mu \alpha }}{2}}}\prod\limits_{k = 1}^K {{q_k^{\frac{{\mu \alpha  - \phi }}{2}{a_k} - \frac{{\mu \alpha }}{2}}}{B^{{a_k}}}{C^{1 - {a_k}}}} } \notag\\
&\times\frac{1}{{2\pi \rm i}}\int_{\rm{c_1} - \rm i\infty }^{{{\rm{c}}_1}{\rm{ + i}}\infty } {{{\left( {{2^{ - R}}} \right)}^t}\frac{{\Gamma \left( { - t} \right)}}{{\Gamma \left( {1 - t} \right)}}}\prod\limits_{k = 1}^K {{{\left( {\frac{{\Gamma \left( { - t - \frac{\phi }{2}} \right)}}{{\Gamma \left( { - t} \right)}}} \right)}^{{a_k}}}{{\left( {\frac{{\Gamma \left( { - t - \frac{{\mu \alpha }}{2}} \right)}}{{\Gamma \left( { - t} \right)}}} \right)}^{1 - {a_k}}}} dt,
\end{align}
\hrulefill
\end{figure*}
where ${\bf{\Omega }} = \left\{ { {\left( {{a_1}, \cdots ,{a_K}} \right)} : {a_k} = \left\{ {0,1} \right\},1 \le k \le K} \right\}$. To obtain the asymptotic outage expression as $\hat \rho\to \infty$, it suffices to consider the dominant terms ${{{\hat \rho }^{ {\frac{{\mu \alpha  - \phi }}{2}}\sum\nolimits_{k = 1}^K {{a_k}}  - \frac{{K\mu \alpha }}{2}}}}$ in \eqref{eqn:IR_900}. Thus, it amounts to finding the minimum value of the corresponding exponent ${ {\frac{{\mu \alpha  - \phi }}{2}}\sum\nolimits_{k = 1}^K {{a_k}}  - \frac{{K\mu \alpha }}{2}}$. Apparently, the minimum value of the exponent depends on the sign of ${\mu \alpha  - \phi }$. Specifically, the minimum of the exponent is attained with $a_1=\cdots=a_K=1$ if ${\mu \alpha  - \phi }>0$ and $a_1=\cdots=a_K=0$ otherwise.
Thus the outage probability is asymptotically expressed as \eqref{eqn:IR_101}, as shown at the top of the next page.
\begin{figure*}[!t]
\begin{align}\label{eqn:IR_101}
{p_{out,K}} \simeq \left\{ \begin{gathered}
  {\left( {\frac{\phi }{{2\Gamma (\mu )}}} \right)^K}{{\hat \rho }^{\frac{{ - \phi K}}{2}}}\prod\limits_{k = 1}^K {q_k^{\frac{{ - \phi }}{2}}B\frac{1}{{2\pi {\text{i}}}}\int_{{{\text{c}}_1} - {\text{i}}\infty }^{{{\text{c}}_1} + {\text{i}}\infty } {{{\left( {{2^{ - R}}} \right)}^t}\frac{{\Gamma \left( { - t} \right)}}{{\Gamma \left( {1 - t} \right)}}\prod\limits_{k = 1}^K {\frac{{\Gamma \left( { - t - \frac{\phi }{2}} \right)}}{{\Gamma \left( { - t} \right)}}} dt,\mu \alpha  - \phi  > 0} }  \hfill \\
  {\left( {\frac{\phi }{{2\Gamma (\mu )}}} \right)^K}{{\hat \rho }^{\frac{{ - \mu \alpha K}}{2}}}\prod\limits_{k = 1}^K {q_k^{\frac{{ - \mu \alpha }}{2}}C\frac{1}{{2\pi {\text{i}}}}\int_{{{\text{c}}_1} - {\text{i}}\infty }^{{{\text{c}}_1} + {\text{i}}\infty } {{{\left( {{2^{ - R}}} \right)}^t}\frac{{\Gamma \left( { - t} \right)}}{{\Gamma \left( {1 - t} \right)}}\prod\limits_{k = 1}^K {\frac{{\Gamma \left( { - t - \frac{{\mu \alpha }}{2}} \right)}}{{\Gamma \left( { - t} \right)}}} dt,\mu \alpha  - \phi  < 0} }  \hfill \\
\end{gathered}  \right.
\end{align}
\hrulefill
\end{figure*}
By using the residue theorem and identifying the integral in \eqref{eqn:IR_101} with Meijer's G-function \cite[Eq. 9.301]{gradshteyn2014table} along with \cite[Eq. 9.312]{gradshteyn2014table}, the asymptotic outage probability $p_{out,K}^{asy}$ of HARQ-IR-aided THz communications can be derived as \eqref{eqn:IR_1000}.

\section{Proof of Remark \ref{eqn:mh_do}}\label{sec:le3}

The asymptotic outage probabilities of multi-hop HARQ-IR-aided THz communications in the cases of $P_B > 0$ and $P_B = 0$ are separately discussed in the following.
\subsubsection{$P_B > 0$}
If $P_B > 0$, the occurrence of the outage event in the high SNR regime, i.e., $\hat \rho \to \infty$, only depends on whether any hop is blocked or not. Accordingly, ${p_{suc,k}^{\left( L \right)}} = {P_N}^L$ if $k=L$ and zero otherwise. With \eqref{eqn:MH_1}, the outage probability is asymptotic to \eqref{mh_do2}.
\subsubsection{$P_B = 0$}
By noticing that the outage event could occur at any hop, using the law of total probability yields the outage expression of the multi-hop HARQ-IR-aided THz communications as \eqref{eqn:le301}, as shown at the top of the next page, where the term $\chi _n$ represents the probability of the event that the transmissions in the first $n-1$ hops are successful but the outage takes place during the $n$-th hop. If there is no blockage during each hop, i.e., $P_B = 0$, according to Theorem \ref{eqn:le_2}, the outage probability ${p_{out,{k}}^l}$ is proportional to ${{\hat \rho }^{ - \frac{{ {k }\min \{ \phi ,\alpha \mu \} }}{2}}}$, i.e., $p_{out,k}^l \propto {{\hat \rho }^{ - \frac{{k\min \{ \phi ,\alpha \mu \} }}{2}}}$. Besides, by using the identity $p_{suc,k}^l = p_{out,k - 1}^l - p_{out,k}^l$, we have $p_{suc,k}^l  \propto {{\hat \rho }^{ - \frac{{\left( {k - 1} \right)\min \{ \phi ,\alpha \mu \} }}{2}}}$. By substituting the asymptotic forms of $p_{out,k}^l$ and $p_{suc,k}^l$ into \eqref{eqn:le301}, applying the dominant term approximation to \eqref{eqn:le301} leads to $p_{out,K}^{\left( L \right)} \simeq {\chi _L} \propto {{\hat \rho }^{ - \frac{{\left( {K - L + 1} \right)\min \{ \phi ,\alpha \mu \} }}{2}}}$. As a consequence, the diversity order can be derived as \eqref{mh_do3}.
\begin{figure*}[!t]
\begin{align}\label{eqn:le301}
p_{out,K}^{\left( L \right)} =& \underbrace {p_{out,K}^1}_{{\chi _1}} + \underbrace {\left( {p_{suc,1}^1p_{out,K - 1}^2 + p_{suc,2}^1p_{out,K - 2}^2 +  \cdots  + p_{suc,K - 1}^1p_{out,1}^2} \right)}_{{\chi _2}}\notag\\
& \cdots  + \underbrace {\sum\limits_{\scriptstyle\sum\nolimits_{l = 1}^n {{\kappa _l}}  = K\hfill\atop
\scriptstyle{\kappa _1}, \cdots ,{\kappa _n} \ge 1\hfill} {\prod\limits_{l = 1}^{n - 1} {p_{suc,{\kappa _l}}^l} p_{out,{\kappa _n}}^n} }_{{\chi _n}} +  \cdots  + \underbrace {\sum\limits_{\scriptstyle\sum\nolimits_{l = 1}^L {{\kappa _l}}  = K\hfill\atop
\scriptstyle{\kappa _1}, \cdots ,{\kappa _L} \ge 1\hfill} {\prod\limits_{l = 1}^{L - 1} {p_{suc,{\kappa _l}}^l} p_{out,{\kappa _L}}^L} }_{{\chi _L}},
\end{align}
\hrulefill
\end{figure*}


\bibliographystyle{ieeetran}
\bibliography{THz}
\end{document}